\documentclass[acmsmall,screen]{acmart} 
\PassOptionsToPackage{table,x11names}{xcolor}
\setcopyright{none}
\acmJournal{TQC}

\usepackage[capitalise,noabbrev]{cleveref}
\usepackage{wrapfig}
\usepackage[inline]{enumitem}
\usepackage{extarrows}
\usepackage{longtable}
\usepackage{adjustbox}
\usepackage{url}
\usepackage{multicol}
\usepackage{stmaryrd}
\usepackage{proof}
\usepackage{bbold}
\usepackage[utf8]{inputenc}
\usepackage{newunicodechar}
\usepackage{underoverlap}
\usepackage{yquant}
\usepackage[caption=false]{subfig}
\usepackage{comment}
\usepackage[frozencache=true,cachedir=minted-cache]{minted}
\newcommand{\agda}[1]{\mintinline{agda}{#1}}
\newcommand{\agdafrag}[3]{\inputminted[fontsize=\footnotesize,firstline=#2,lastline=#3]{agda}{agda/#1}}

\usepackage{soul}

\usepackage{amssymb}
\usepackage{newtxmath}

\usepackage{tikz}
\usepackage{quiver}

\usetikzlibrary{decorations.markings}
\usetikzlibrary{quotes,fit,positioning}
\usetikzlibrary{arrows.meta, tikzmark}
\usetikzlibrary{knots}
\tikzstyle{func}=[rectangle,draw,fill=black!20,minimum size=1.9em,text width=2.4em, text centered]
\usetikzlibrary{braids}
\tikzset{>=latex}
\tikzstyle{blackdot}=[circle, draw=black, fill=black!75, inner sep=.4ex, line width=.7pt]
\tikzstyle{whitedot}=[circle, draw=black, fill=white, inner sep=.4ex, line width=.7pt]

\newcommand\tinymatrix[1]{\left( \renewcommand\thickspace{\kern2pt} \scriptstyle\begin{smallmatrix} #1 \end{smallmatrix} \right)}

\newtheorem{theorem}{Theorem}
\newtheorem{corollary}[theorem]{Corollary}
\newtheorem{lemma}[theorem]{Lemma}
\theoremstyle{definition}
\newtheorem{definition}[theorem]{Definition}

\usepackage{macros}
\usepackage{syntax}

\newunicodechar{ˡ}{${}^{l}$}
\newunicodechar{ʳ}{${}^{r}$}
\newunicodechar{Π}{$\Pi$}
\newunicodechar{ℕ}{$\mathbb{N}$}
\newunicodechar{⊡}{$\boxdot$}
\newunicodechar{≡}{$\equiv$}
\newunicodechar{⟦}{$\llbracket$}
\newunicodechar{⟧}{$\rrbracket$}
\newunicodechar{𝟚}{$\mathbb{2}$}
\newunicodechar{ϕ}{$\phi$}
\newunicodechar{ᵤ}{$_{u}$}
\newunicodechar{₀}{$_{0}$}
\newunicodechar{₁}{$_{1}$}
\newunicodechar{₂}{$_{2}$}
\newunicodechar{⇔}{$\Leftrightarrow$}
\newunicodechar{₃}{$_{3}$}
\newunicodechar{₄}{$_{4}$}
\newunicodechar{₅}{$_{5}$}
\newunicodechar{₆}{$_{6}$}
\newunicodechar{∘}{$\circ$}
\newunicodechar{⟷}{$\leftrightarrow$}
\newunicodechar{◎}{$\fatsemi$} 
\newunicodechar{⊗}{$\otimes$}
\newunicodechar{₊}{${}_{+}$}
\newunicodechar{⊕}{$\oplus$}
\newunicodechar{λ}{$\lambda$}
\newunicodechar{⁻}{${}^{-}$}
\newunicodechar{⊚}{$@$}
\newunicodechar{π}{$\pi$}
\newunicodechar{∀}{$\forall$}
\newunicodechar{⋆}{$\star$}
\newunicodechar{∎}{$\blacksquare$}
\newunicodechar{𝔽}{$\mathbb{F}$}
\newunicodechar{𝕋}{$\mathbb{T}$}
\newunicodechar{◯}{$\odot$} 

\begin{document}

\title{The Quantum Effect}
\subtitle{A Recipe for \qpi}

\author{Jacques Carette}
\orcid{}
\affiliation{
  \department{}
  \institution{McMaster University}
  \city{Hamilton}
  \postcode{}
  \country{Canada}
}
\email{}

\author{Chris Heunen}
\orcid{0000-0001-7393-2640}
\affiliation{
  \department{School of Informatics}
  \institution{University of Edinburgh}
  \city{Edinburgh}
  \postcode{EH8 9AB}
  \country{United Kingdom}
}
\email{chris.heunen@ed.ac.uk}

\author{Robin Kaarsgaard}
\orcid{0000-0002-7672-799X}
\affiliation{
  \department{Department of Mathematics and Computer Science}
  \institution{University of Southern Denmark}
  \city{Odense}
  \postcode{5230}
  \country{Denmark}
}
\email{kaarsgaard@imada.sdu.dk}

\author{Amr Sabry}
\orcid{0000-0002-1025-7331}
\affiliation{
  \department{Department of Computer Science}
  \institution{Indiana University}
  \city{Bloomington}
  \postcode{47408}
  \country{USA}
}
\email{sabry@indiana.edu}

\begin{abstract}
Free categorical constructions characterise quantum
computing as the combination of two copies of a reversible classical
model, glued by the complementarity equations of classical
structures. This recipe effectively constructs a computationally
universal quantum programming language from two copies of $\PiLang$, the
internal language of rig groupoids. The construction consists of Hughes' arrows. Thus we answer positively the question whether a computational effect exists that turns reversible classical computation into quantum computation: the quantum effect. Measurements can be added by 
layering a further effect on top. Our construction also enables some reasoning about
quantum programs (with or without measurement) through a combination
of classical reasoning and reasoning about complementarity.
\end{abstract}

\maketitle

\section{Introduction}

A distinguishing and well-established aspect of quantum
mechanics is the concept of \emph{complementarity}. Roughly speaking
quantum data is such that an observation in one experimental
setting excludes the possibility of gaining any information in
\emph{complementary} settings. This experimental
characterisation has been formalised as an equation that relates a
particular collection of ``classical morphisms'' in the category of
finite-dimensional Hilbert spaces, and has inspired the development of
axiomatic models of quantum theory~\cite{coeckeduncan:zx}.

From a programming language perspective, one thus expects that a
particular collection of ``classical subexpressions'' in any universal
quantum programming language exhibits this complementarity property. 
More foundationally, and of practical significance, is the opposite
question: is it possible to induce quantum
behaviour by imposing the complementarity property on two
classical structures, each modeled by a classical reversible
programming language? The most important practical significance of
such a construction would be that some forms of reasoning about quantum programs reduce to classical reasoning augmented with one complementarity
equation.  Additionally, the design of quantum programming languages
itself would reduce to designing two appropriate
classical languages whose interaction is constrained by
complementarity. Foundationally, this construction would turn around
the prevalent view of quantum computing~\cite{GLRSV2013-pldi,SV2005-qlambda,SV2009-qlambdabook,voichick2022qunity},
potentially shedding some light on a long-standing foundational
question in physics about the relationship between quantum and
classical theories~\cite{phys-foundations}.

This article gives just such a recipe, namely constructing a computationally universal quantum
programming language from two copies of a (particlar) universal classical
reversible language. The mathematical formalism is expressed
using free categorical constructions, and makes heavy use of Hughes' arrows~\cite{hughes:arrows,jacobsheunenhasuo:arrows}. To
also demonstrate the recipe in action and explore its pragmatics in
programming and reasoning about quantum circuits, we apply it to a
canonical reversible programming language \PiLang\ yielding the
computationally universal quantum programming language \qpi, and
implement the entire project in Agda\footnote{Available from the
repository \url{https://github.com/JacquesCarette/QuantumPi}.}

\paragraph*{Related work.}
We discuss related work in detail where appropriate, and provide just
a brief summary here. Quantum programming
languages~\cite{voichick2022qunity, bichsel2020silq, GLRSV2013-pldi,
  SV2009-qlambdabook, sabryvalironvizzotto:symmetric, selinger:qpl,
  paykinrandzdancewic:qwire}, classical reversible programming
languages~\cite{jamessabry:infeff, 10.1145/3498667,
  10.1007/978-3-662-49498-1, jacobsenkaarsgaardthomsen:corefun,
  yokoyamaglueck:janus, yokoyamaaxelsenglueck:rfun}, their categorical
semantics~\cite{pechouxperdrixrennelazamdzhiev:causality,
  chowesterbaan:vonneumann, rennelastaton:enriched,
  westerbaan:kleisli,10.1145/3498687, heunenkaarsgaard:qie,
  heunenkaarsgaardkarvonen:arrows,
  glueckkaarsgaardyokoyama:reversible}, complementarity of classical
structures~\cite{coeckeduncan:zx,coeckeetal:strongcomplementarity},
and amalgamation of categories~\cite{macdonaldscull:amalg} have all
been individually studied before. In particular, complementarity is
central to the ZX-~\cite{coeckeduncan:zx} and
ZH-calculi~\cite{backenskissinger:zh,comfort:tofH}, though more axioms are needed for completeness. Given this wealth
of prior work, our main contribution is to design an infrastructure in
which these established ideas can be organised in a way for quantum
behaviour to emerge from classical programming languages using
computational effects in the form of arrows. 
The quest for such a computational `quantum effect' also underlies~\cite{altenkirchgreen:quantumio}, and the current article vastly improves on that very early work.

\begin{figure}[t]
\begin{center}
\begin{tikzpicture}[scale=0.7, every node/.style={scale=0.7}]

\node [circle,draw] (pi) at (0,0) {$\PiLang$};
\node [circle,draw,text=cyan] (piz) at (-1.5,1.5) {$\PizLang$};
\node [circle,draw,text=magenta] (pih) at (1.5,1.5) {$\PihLang$};
\node [circle,draw] (pizh) at (0,3.5) {$\PizhLang$};
\node [circle,draw] (spizhe) at (0,6.5) {$\SPizheLang$};
\node [circle,draw] (spizheq) at (0,10) {$\qpi$};

\draw[-stealth] (pi) to node[left=2mm]{} (piz);
\draw[-stealth] (pi) to node[right=2mm]{} (pih);
\draw[-stealth] (piz) to node[left]{} (pizh);
\draw[-stealth] (pih) to node[right]{} (pizh);
\draw[-stealth] (pizh) to node[left]{Arrow} (spizhe);
\draw[-stealth] (spizhe) to node[left] {Complementarity} (spizheq);

\node[draw,text width=2cm, below=0.1cm of pi]
     {$\pid; \swapp; \ldots$};
\node[below=0.2cm of pizh] {Arrows};
\node[draw,text width=2cm, left=0.1cm of piz]
     {$\zlang{\pid}; \zlang{\swapp}; \ldots$};
\node[draw,text width=2cm, right=0.1cm of pih]
     {$\hlang{\pid}; \hlang{\swapp}; \ldots$};
\node[draw,text width=3.4cm,left=0.3cm of pizh]
     {Fig.~\ref{fig:pizh} new constructs: \\
       $[\zlang{c_1}, \hlang{c_2}, \zlang{c_3}, \hlang{c_4}, \ldots]$};
\node[draw,text width=3.4cm,left=0.3cm of spizhe]
     {Fig.~\ref{fig:spizhe} new constructs: \\
       $\cm{lift}~xs$};
\end{tikzpicture}
\begin{tikzpicture}[scale=0.7, every node/.style={scale=0.7}]

\node (pi) at (0,0) {Rig groupoid};
\node [text=cyan] (piz) at (-1.5,1.5) {$\Unitary$};
\node [text=magenta] (pih) at (1.5,1.5) {$\Aut_{R\phid}(\Unitary)$};
\node (pizh) at (0,3.5) {$\SymMonAmalg(\zlang{\Unitary}, \hlang{\Aut_{R\phid}(\Unitary))}$};
\node (spizhe) at (0,6.5) {$\SymMonAmalg(\zlang{\Unitary}, \hlang{\Aut_{R\phid}(\Unitary))}_{I\oplus I}$};
\node (spizheq) at (0,10) {$\cat{Contraction}$};

\node[text width=2.8cm, below=0.2cm of pi] {};

\draw[-stealth] (pi) to node[left=2mm]{} (piz);
\draw[-stealth] (pi) to node[right=2mm]{} (pih);
\draw[-stealth] (piz) to node[left]{} (pizh);
\draw[-stealth] (pih) to node[right]{} (pizh);
\draw[-stealth] (pizh) to node[right]{} (spizhe);
\draw[-stealth] (spizhe) to node[right]{} (spizheq);
\end{tikzpicture}\end{center}
\caption{\label{fig:struct}General overview of the approach to
  computationally universal quantum computation. The left diagram
  shows the progression of languages and the right diagram shows the
  corresponding progression of categories. We start at the bottom with
  a classical reversible language universal for finite types, that take semantics in rig groupoids. Moving up, we make two copies of the
  language that are rotated with respect to each other and embed them
  in the category $\Unitary$. The two copies are amalgamated in a
  first effectful arrow layer that provides interleaving constructs. A
  second arrow layer is then added in which classical information can
  be created and deleted. Once the complementarity equation is imposed
  on the final language at the top, we get a computationally universal
  quantum programming language, with semantics in the category $\cat{Contraction}$ of finite-dimensional Hilbert spaces and linear contractions. }
\end{figure}

\paragraph*{Outline.}
Fig.~\ref{fig:struct} provides a summary of our technical development,
which proceeds in two parallel threads: the design of a
computationally universal quantum programming language on the left and
the corresponding categorical models on the right. This summary is
expanded in Sec.~\ref{sec:story}, which presents the resulting
programming language and its model without exposing or explaining the
intermediate steps. The main point is for the reader to appreciate the
main technical result --- the canonicity theorem --- without the
distraction of some of the more involved technical properties and
their proofs. The full technical development starts in
Sec.~\ref{sec:background} with a review of relevant background on the
categorical semantics of quantum computing. Sec.~\ref{sec:pi} then
reviews the core classical reversible programming language
\PiLang\ and its semantics in rig groupoids. The first step of our
construction in Sec.~\ref{sec:automorphisms} is to take two copies of
\PiLang, called \PizLang\ and \PihLang, and embeds their semantics, in
two different ways, in the category $\Unitary$ of finite-dimensional
Hilbert spaces and unitaries. These two languages are then amalgamated
in Sec.~\ref{sec:amalgamation} to produce a language \PizhLang\ in
which expressions from \PizLang\ and \PihLang\ can be freely
interleaved. The amalgamated language is then extended in
Sec.~\ref{sec:spice} to the language \SPizheLang\ which exposes the
classical structures explicitly. Sec.~\ref{sec:canonicity} proves our
main result, the canonicity theorem. Sec.~\ref{sec:reasoning}
introduces \qpi, the user-level interface of \SPizheLang. Unlike the
previous section, the development uses our Agda specification to
provide executable circuits \emph{and} machine checkable proofs of
various circuit equivalences. This section demonstrates that some forms of reasoning
about quantum programs in \qpi\ indeed reduces to classical reasoning
augmented with the complementarity equation. It also shows that
\qpi\ can model gates with complex numbers and be extended with a
measurement postulate to model complete quantum algorithms. The
article concludes with a section summarising the results and
discussing possible future directions. 

\section{Three Arrows to Quantum}
\label{sec:story}

We first present a complete, somewhat informal, overview of
our main result.  The necessary categorical constructions, models, and
their properties that are needed to justify the result are quite
involved and form the main body of the paper.

Usually computational effects are introduced as a layer above a pure language using
monads or arrows. Here, we take the unusual step of using \emph{two}
copies of a core language and introduce an arrow that amalgamates
them, allowing arbitrary interleaving of expressions from either
language. As will be explained in the remainder of the paper, the core
language is $\PiLang$, a universal language for classical reversible
circuits providing facilities for sequential, additive, and
multiplicative compositions of circuits. The language $\PiLang$
includes many constructs, the most important of which for this section
is $\swapp : \twot \fromto \twot$ where $\twot$ is the
type containing two values, namely a left injection and a right injection;
the semantics of $\swapp$ is an automorphism of the type $\twot$
that exchanges the left and right injections.  Distributivity of
the multiplicative structure over the additive one allows the language
to express conditional (or controlled) execution via an expression
$\ctrlgate~c$ which takes a pair of inputs, one of which is of type
$\twot$ and executes $c$ only if the given value of type
$\twot$ is the one denoting true.

The two copies of $\PiLang$, called $\PizLang$ and~$\PihLang$, are
amalgamated into a larger language via their own constructor $\cm{arr}_Z$ and
$\cm{arr}_{\phi}$. We think of $\phi$ as an angle which specifies how much
one language is ``rotated'' with respect to the other. The rotations for
general expressions are generated in a compositional type-directed
manner from the rotation on $\swapp$. As a consequence, the two arrows
share the same multiplicative structure, whose semantics is identical
in each of the sublanguages and in their amalgamation. The semantics
of the additive structure is however not lifted to the amalgamated
language. Finally, a third arrow, which also shares the same
multiplicative structure, is layered above the amalgamated language
allowing the introduction of partially reversible maps for a state
$\cm{zero}$ and an effect $\cm{assertZero}$.

Collecting these ideas, given a classical reversible language whose
expressions are denoted by $c$, the syntax of \qpi\ is:
\begin{align*}
  d & \defeq \cm{arr}_Z~c \mid \cm{arr}_{\phi}~c \\
  & \mid \unitetl \mid \unititl \mid \swapt \mid \assoct \mid \associt \\
  & \mid \pid \mid d \ggg d \mid d \arrprod d \mid
    \inv~d \mid \cm{zero} \mid \cm{assertZero}
\end{align*}
The first line lifts the underlying $c$ expressions using two
different arrows. The next line provides a multiplicative monoidal
structure. The last line provides identity, sequential and parallel
compositions, a (partial) inverse, and the two state and effect
constants. We will additionally use the shorthands
\begin{align*}
  \cm{one} &= \cm{zero} \ggg \cm{arr}_Z~\swapp \\
  \cm{assertOne} &= \cm{arr}_Z~\swapp \ggg \cm{assertZero}
\end{align*}
but we stress that these are merely syntactic conveniences and not
part of \qpi as such.

The intention is for the underlying languages to be lifted
in such a way that $\cm{arr}_\phi~\swapp$ can be thought of as a
change of basis between a pair of complementary orthonormal bases. Let
$\{F_Z,T_Z\}$ be the denotations of the two values of $\twot$
lifted via $\cm{arr}_Z$, and let $\{F_X,T_X\}$ be these values applied
to $\cm{arr}_\phi~\swapp$ (note that this is strictly different from the
denotations of the two values of $\twot$ lifted via
$\cm{arr}_\phi$).  We require that the denotation of $\cm{zero}$ be
$F_Z$ which is the one value recognized by the predicate
$\cm{assertZero}$. Formally, the semantics satisfies the following
execution equations:
\begin{align*}
  \cm{zero} \ggg \cm{assertZero} &~=~ \pid \\
  (\cm{zero} \arrprod \pid) \ggg \cm{arr}_Z~(\ctrlgate~c) &~=~ \cm{zero} \arrprod \pid \\
  (\cm{one} \arrprod \pid) \ggg \cm{arr}_Z~(\ctrlgate~c) &~=~ \cm{one} \arrprod \cm{arr}_Z~c \\
  \cm{zero} \ggg \cm{arr}_\phi~\swapp \ggg \cm{assertOne} &~=~ \cm{one} \ggg \cm{arr}_\phi~\swapp \ggg \cm{assertZero}
\end{align*}
\noindent The first equation states that $\cm{assertZero}$ is the
partial inverse of $\cm{zero}$. The second and third equation ensure
that $F_Z$ behaves as false in conditional expressions and that $T_Z$
behaves as true, while the final equation describes a relationship the
states and effects of the two bases. Later, we will see this forces
the meaning of $\cm{arr}_\phi~\swapp$ to be a real matrix.

\begin{figure}
  \begin{align*}
    \begin{aligned}\begin{tikzpicture}[scale=.5, every node/.style={scale=0.75}]
      \node[whitedot] (t) at (0,1) {};
      \node[whitedot] (b) at (1,-.5) {};
      \draw (t.center) to[out=180,in=-90] +(-.75,1.5);
      \draw (t.center) to[out=0,in=-90] +(.75,1.5);
      \draw (b.center) to[out=180,in=-90] (t.south);
      \draw (b.center) to +(0,-1);
      \draw (b.center) to[out=0,in=-90] +(1.5,3);
      \node[whitedot,fill=cyan] (t) at (0,1) {};
      \node[whitedot,fill=cyan] (b) at (1,-.5) {};
      \node[below left] at (t) {$\cm{copy}_Z$};
      \node[below left] at (b) {$\cm{copy}_Z$};
    \end{tikzpicture}\end{aligned}
    &=
    \begin{aligned}\begin{tikzpicture}[yscale=.5,xscale=-.5,every node/.style={scale=0.75}]
      \node[whitedot] (t) at (0,1) {};
      \node[whitedot] (b) at (1,-.5) {};
      \draw (t.center) to[out=180,in=-90] +(-.75,1.5);
      \draw (t.center) to[out=0,in=-90] +(.75,1.5);
      \draw (b.center) to[out=180,in=-90] (t.south);
      \draw (b.center) to +(0,-1);
      \draw (b.center) to[out=0,in=-90] +(1.5,3);
      \node[whitedot,fill=cyan] (t) at (0,1) {};
      \node[whitedot,fill=cyan] (b) at (1,-.5) {};
      \node[below right] at (t) {$\cm{copy}_Z$};
      \node[below right] at (b) {$\cm{copy}_Z$};
    \end{tikzpicture}\end{aligned}
    &
  \begin{aligned}\begin{tikzpicture}[scale=.5,every node/.style={scale=0.75}]
    \node[whitedot] (d) {};
    \draw (d.center) to +(0,-1.5);
    \draw (d.center) to[out=0,in=-90] +(1,2.5);
    \draw (d.center) to[out=180,in=-90] +(-1,2.5);
    \node[whitedot,fill=cyan] (d) {};
      \node[below left] at (d) {$\cm{copy}_Z$};
  \end{tikzpicture}\end{aligned}
  &=
  \begin{aligned}\begin{tikzpicture}[scale=.5,every node/.style={scale=0.75}]
    \node[whitedot,fill=cyan] (d) {};
    \draw (d.center) to +(0,-1.5);
    \draw (d.center) to[out=0,in=-90] ++(1,1) to[out=90,in=-90] ++(-2,1.5);
    \draw (d.center) to[out=180,in=-90] ++(-1,1) to[out=90,in=-90] ++(2,1.5);
    \node[whitedot,fill=cyan] (d) {};
      \node[below right] at (d) {$\cm{copy}_Z$};
  \end{tikzpicture}\end{aligned}
  \\
  \begin{aligned}\begin{tikzpicture}[every node/.style={scale=0.75}]
      \node (0) at (0,0.33) {};
      \node [whitedot] (1) at (0,1) {};
      \node [whitedot] (2) at (0,2) {};
      \node (3) at (0,2.66) {};
      \draw  (0) to (1.center);
      \draw  (2.center) to (3);
      \draw [in=left, out=left, looseness=1.5] (1.center) to (2.center);
      \draw [in=right, out=right, looseness=1.5] (1.center) to (2.center);
      \node [whitedot,fill=cyan] (1) at (0,1) {};
      \node [whitedot,fill=cyan] (2) at (0,2) {};
      \node[below left] at (1) {$\cm{copy}_Z$};
      \node[above left] at (2) {$\inv~{\cm{copy}_Z}$};
  \end{tikzpicture}\end{aligned}
  &= 
  \begin{aligned}\begin{tikzpicture}
      \node (0) at (0,0.33) {};
      \node (1) at (0,2.66) {};
      \draw  (0) to (1);
  \end{tikzpicture}\end{aligned}
    &
  \begin{aligned}\begin{tikzpicture}[scale=0.75,every node/.style={scale=0.75}]
        \node (0) at (0,0) {};
        \node (0a) at (0,1) {};
        \node [whitedot] (1) at (0.5,2) {};
        \node [whitedot] (2) at (1.5,1) {};
        \node (3) at (1.5,0) {};
        \node (4) at (2,3) {};
        \node (4a) at (2,2) {};
        \node (5) at (0.5,3) {};
        \draw (0) to (0a.center);
        \draw [out=90, in=180] (0a.center) to (1.center);
        \draw [out=0, in=180] (1.center) to (2.center);
        \draw [out=0, in=270] (2.center) to (4a.center);
        \draw (4a.center) to (4);
        \draw (2.center) to (3);
        \draw (1.center) to (5);
        \node [whitedot,fill=cyan] (1) at (0.5,2) {};
        \node [whitedot,fill=cyan] (2) at (1.5,1) {};
      \node[above left] at (1) {$\inv~\cm{copy}_Z$};
        \node[below right] at (2) {$\cm{copy}_Z$};
  \end{tikzpicture}\end{aligned}
  &=
  \begin{aligned}\begin{tikzpicture}[yscale=0.75,xscale=-0.75,every node/.style={scale=0.75}]
        \node (0) at (0,0) {}; 
        \node (0a) at (0,1) {};
        \node [whitedot] (1) at (0.5,2) {};
        \node [whitedot] (2) at (1.5,1) {};
        \node (3) at (1.5,0) {};
        \node (4) at (2,3) {};
        \node (4a) at (2,2) {};
        \node (5) at (0.5,3) {};
        \draw (0) to (0a.center);
        \draw [out=90, in=180] (0a.center) to (1.center);
        \draw [out=0, in=180] (1.center) to (2.center);
        \draw [out=0, in=270] (2.center) to (4a.center);
        \draw (4a.center) to (4);
        \draw (2.center) to (3);
        \draw (1.center) to (5);
        \node [whitedot,fill=cyan] (1) at (0.5,2) {};
        \node [whitedot,fill=cyan] (2) at (1.5,1) {};
      \node[above right] at (1) {$\inv~\cm{copy}_Z$};
        \node[below left] at (2) {$\cm{copy}_Z$};
  \end{tikzpicture}\end{aligned}
  \end{align*}
  \caption{Depictions of the properties for the classical structure
    for $\cm{copy}_Z$. The laws for $\cm{copy}_X$ are
    identical. These diagrams are provided only as illustration; the
    graphical language of monoidal categories~\cite{heunenvicary:cqt}
    will not be used in the rest of this article.}
  \label{fig:stringdiagrams1}
\end{figure}  

With the introduction of state $\cm{zero}$ and effect
$\cm{assertZero}$, the arrow layer can express two different
computations that copy and merge expressions of type $\twot$:
\begin{align*}
\cm{copy}_Z &~=~
    (\unititl \ggg \swapt) \ggg (\pid \arrprod \cm{zero}) \ggg \cm{arr}_Z~(\ctrlgate~\swapp) \\
    \cm{copy}_X &~=~ \cm{arr}_{\phi}~\swapp \ggg \cm{copy}_Z \ggg
    (\cm{arr}_{\phi}~\swapp \arrprod \cm{arr}_{\phi}~\swapp)
\end{align*}

\noindent The expression $\cm{copy}_Z$ takes a value $b$ which is
either $F_Z$ or $T_Z$ and duplicates it. The expression $\cm{copy}_X$
rotates the values of type $\twot$ before and after invoking
$\cm{copy}_Z$ which has the result of duplicating $F_X$ or
$T_X$. Intuitively, the inverse copy maps merge two equal values and
are undefined otherwise. Formally, each $\cm{copy}$ and its inverse
form a classical structure that satisfies the laws depicted in
Fig.~\ref{fig:stringdiagrams1}.  These properties state that the
duplicated values are indistinguishable, that the duplication can
happen in any order, and that immediately applying the inverse
perfectly undoes the effect of copying.

\begin{wrapfigure}{r}{0.5\textwidth}
\begin{center}\begin{tikzpicture}[scale=0.5, every node/.style={scale=0.5}, outer sep=3pt]
  \tikzstyle{box}=[fill=cyan!30, rounded corners=3pt];
  \tikzstyle{markedbox}=[fill=magenta!30, rounded corners=3pt];
  \node[box] (A) at (0,0) {$F_Z$};
  \node[box] (B) at (6,0) {$F_X$};
  \node[box] (C) at (-2,2) {$F_Z$};
  \node[box] (D) at (2,2) {$F_Z$};
  \node[box] (E) at (6,2) {$F_X$};
  \node[markedbox] (F) at (4,4) {$F_X$};
  \node[box] (G) at (-2,6) {$F_Z$};
  \node[box] (H) at (2,6) {$F_X$};
  \node[box] (I) at (6,6) {$F_X$};
  \node[markedbox] (J) at (0,8) {$F_Z$};
  \node[box] (K) at (6,8) {$F_X$};
  \draw[-stealth] (A) to[out=150,in=-60] node[right=2mm]{$\cm{copy}_Z$}(C);
  \draw[-stealth] (A) to[out=30,in=-120] (D);
  \draw[-stealth] (B) to (E);
  \draw[-stealth] (D) to[out=60,in=-150] node[right=0.2mm]{$\inv~\cm{copy}_X$}(F);
  \draw[-stealth] (E) to[out=120,in=-30] (F);
  \draw[-stealth] (F) to[out=150,in=-60] node[right=2mm]{$\cm{copy}_X$}(H);
  \draw[-stealth] (F) to[out=30,in=-120] (I);
  \draw[-stealth] (H) to[out=120,in=-30] (J);
  \draw[-stealth] (C) to (G);
  \draw[-stealth] (G) to[out=60,in=-150] node[right=0.5mm]{$\inv~\cm{copy}_Z$}(J);
  \draw[-stealth] (I) to (K);
\end{tikzpicture}\end{center}
\end{wrapfigure}
Recall however that the amalgamated language allows arbitrary
interleaving of expressions from either copy of $\PiLang$. In other
words, it is possible to apply $\cm{copy}_Z$ to $F_X$ or the inverse
of $\cm{copy}_X$ to $(F_Z,F_X)$. Specifically, consider the
execution diagram on the right.  By the laws of classical structures
we know that the two copying actions in the execution diagram
duplicate the given value but, a priori, we have no idea what would
happen in the two magenta nodes. However if we require that the entire
diagram behaves as the identify, there is one choice for the marked
values that forces $\inv~\cm{copy}_X~(F_Z,F_X) = F_X$ and
$\inv~\cm{copy}_Z~(F_Z,F_X) = F_Z$.  Formally, this requirement
is expressed by the complementarity equation with the following
graphical representations:

  
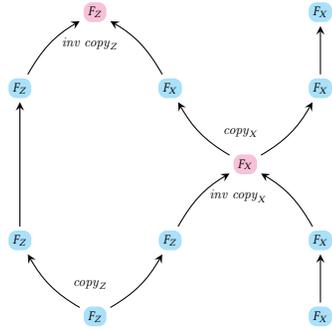
\begin{align*}
    \begin{aligned}\begin{tikzpicture}[scale=.75,every node/.style={scale=0.75}]
      \node (A) at (0,0) {};
      \node (B) at (1.75,0) {};
      \node (b1) [blackdot] at (0,-1) {};
      \node (w1) [whitedot] at (1,-2) {};
      \node (w2) [whitedot] at (1,-3) {};
      \node (b2) [blackdot] at (0,-4) {};
      \node (C) at (0,-5) {};
      \node (D) at (1.75,-5) {};
      \draw (A.center) to (b1.center);
      \draw (b1.center) to [out=right, in=left] (w1.center);
      \draw (w1.center) to (w2.center);
      \draw (w2.center) to [out=left, in=right] (b2.center);
      \draw (b2.center) to (C.center);
      \draw (w2.center) to [out=right, in=up] (D.center);
      \draw (w1.center) to [out=right, in=down] (B.center);
      \draw (b1.center) to [out=left, in=left] (b2.center);
      \node (b1) [whitedot,fill=magenta] at (0,-1) {};
      \node (w1) [whitedot,fill=cyan] at (1,-2) {};
      \node (w2) [whitedot,fill=cyan] at (1,-3) {};
      \node (b2) [whitedot,fill=magenta] at (0,-4) {};
    \node[below right=-0.7mm] at (w1) {$\cm{copy}_Z$};
    \node[above right=-0.7mm] at (w2) {$\inv~\cm{copy}_Z$};
    \node[above left] at (b1) {$\inv~\cm{copy}_X$};
    \node[below left] at (b2) {$\cm{copy}_X$};
    \end{tikzpicture}\end{aligned}
    &\qquad=\qquad\;
    
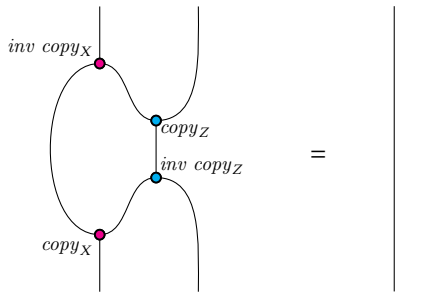
\begin{aligned}\begin{tikzpicture}[scale=.75]
      \draw (0,0) to ++(0,5);
      \draw (1,0) to ++(0,5);
    \end{tikzpicture}\end{aligned}
  \end{align*}

Remarkably, requiring that the semantics of our language satisfies
this equation induces full quantum behaviour: the language becomes
computationally universal for quantum computing. The informal
argument, to be made precise in Thm.~\ref{thm:canonicity}, follows.

\begin{theorem}[Canonicity (informal)]
  If a categorical semantics $\sem{-}$ for our language satisfies the
  classical structure laws, the execution laws, and the
  complementarity law, then the language is computationally universal
  for quantum computing.
\end{theorem}
\begin{proof} (Informal) 
  Without loss of generality, we may assume that $\{F_Z,T_Z\}$ is the
  computational~$Z$ basis and that $\cm{arr}_\phi$ rotates to another
  basis related to the $Z$ basis by $\cm{arr}_\phi~\swapp$. Then the
  laws of classical structures ensure that $\{F_Z,T_Z\}$ and
  $\{F_X,T_X\}$ form orthonormal bases; the execution equations ensure
  that these sets of values are copyable by $\cm{copy}_Z$ and
  $\cm{copy}_X$ respectively and that the change of basis is real; and
  the complementarity equation further ensures that they form a
  complementarity pair of bases for $\mathbb{C}^2$.

  Since $\cm{arr}_\phi~\swapp$ is the symmetry of a monoidal
  structure, it is involutive, and by the above, it is a real change of
  basis between the $Z$ basis and an orthonormal basis complementary
  to it. Hadamard can be characterised, up to conjugation by $X$ and
  $Z$, as the unique real involutive unitary change-of-basis between
  the $Z$ basis and any orthonormal basis complementary to it (see
  Proposition~\ref{prop:charhad}). This yields four possible
  interpretations of $\cm{arr}_\phi~\swapp$, but since all of them are
  real and basis changing, it follows by
  Theorem~\ref{thm:aharonov} that either of them are universal
  in conjunction with the Toffoli gate.
\end{proof}

\section{Categories; Quantum Computing; Computational Universality} 
\label{sec:background}

Category theory deals with abstractions in a uniform and systematic
way, and is widely used to provide compositional programming
semantics. We may think of objects as types, and morphisms as
terms. For the basics we refer to~\citet{leinster:categorytheory}.  To
fix notation and the background knowledge assumed, we briefly discuss
the types of categories that are useful in reversible programming:
dagger categories and rig categories.  Then we will discuss quantum
computing in categorical terms, complementarity, and computational
universality.

\subsection{Dagger Categories and Groupoids}

A morphism $f \colon A \to B$ is \emph{invertible}, or an
\emph{isomorphism}, when there exists a morphism $f^{-1} \colon B \to
A$ such that $f^{-1} \circ f = \id_A$ and $f \circ f^{-1} =
\id_B$. This inverse $f^{-1}$ is necessarily unique. A category where
every morphism is invertible is called a \emph{groupoid}.

At first sight, groupoids form the perfect semantics for reversible
computing. But every step in a computation being reversible is
slightly less restrictive than it being invertible. For each step $f
\colon A \to B$, there must still be a way to `undo' it, given by
$f^\dag \colon B \to A$. This should also still respect composition,
in that $(g \circ f)^\dag = f^\dag \circ g^\dag$ and $\id_A^\dag =
\id_A$. Moreover, a `cancelled undo' should not change anything:
$f^{\dag\dag}=f$. Therefore every morphism $f$ has a partner $f^\dag$.
A category equipped with such a choice of partners is called a
\emph{dagger category}.

A groupoid is an example of a dagger category, where every morphism is \emph{unitary}, that is, $f^\dag = f^{-1}$. Think, for example, of the category $\cat{FinBij}$ with finite sets for objects and bijections for morphisms. But not every dagger category is a groupoid. For example, the dagger category $\cat{PInj}$ has sets as objects, and partial injections as morphisms. Here, the dagger satisfies $f \circ f^\dag \circ f = f$, but not necessarily $f^\dag \circ f = \id$ because $f$ may only be partially defined. In a sense, the dagger category $\cat{PInj}$ is the universal model for reversible computation~\cite{kastl:inverse,heunen:ltwo}.

When a category has a dagger, it makes sense to demand that every other structure on the category respects the dagger, and we will do so. The theory of dagger categories is similar to the theory of categories in some ways, but very different in others~\cite{heunenkarvonen:daggermonads}. 

\subsection{Monoidal Categories and Rig Categories}

Programming becomes easier when the programmer can express more programs natively. For example, it is handy to have type combinators like sums and products. Semantically, this is modelled by considering not mere categories, but monoidal ones. A \emph{monoidal category} is a category equipped with a type combinator that turns two objects $A$ and $B$ into an object $A \otimes B$, and a term combinator that turns two morphisms $f \colon A \to B$ and $f' \colon A' \to B'$ into a morphism $f \otimes f' \colon A \otimes A' \to B \otimes B'$. This has to respect composition and identities. Moreover, there has to be an object $I$ that acts as a unit for~$\otimes$, and isomorphisms $\alpha \colon A \otimes (B \otimes C) \to (A \otimes B) \otimes C$ and $\lambda \colon I \otimes A \to A$ and $\rho \colon A \otimes I \to A$. In a \emph{symmetric monoidal category}, there are additionally isomorphisms $\sigma \colon A \otimes B \to B \otimes A$. All these isomorphisms have to respect composition and satisfy certain coherence conditions, see~\cite[Chapter 1]{heunenvicary:cqt}. We speak of a \emph{(symmetric) monoidal dagger category} when the coherence isomorphisms are unitary.
Intuitively, $g \circ f$ models sequential composition, and $f \otimes g$ models parallel composition. For example, $\cat{FinBij}$ and $\cat{PInj}$ are symmetric monoidal dagger categories under cartesian product.

A \emph{rig category} is monoidal in two ways in a distributive fashion. More precisely, it has two monoidal structures $\oplus$ and $\otimes$, such that $\oplus$ is symmetric monoidal but $\otimes$ not necessarily, and there are isomorphisms $\delta_L \colon A \otimes (B \oplus C) \to (A \otimes B) \oplus (A \otimes C)$ and $\delta_0 \colon A \otimes 0 \to 0$. These isomorphisms again have to respect composition and certain coherence conditions~\cite{laplaza:distributivity}. For example, $\cat{FinBij}$ and $\cat{PInj}$ are not only monoidal under cartesian product, but also under disjoint union, and the appropriate distributivity holds. Intuitively, $f \oplus g$ models a choice between $f$ and $g$.

\subsection{Quantum Computing (categorically)}
\label{sub:quantise}

\newcommand{\hilbspace}{\ensuremath{A}}
\newcommand{\hilbspacep}{\ensuremath{B}}

Quantum computing with pure states is a specific kind of reversible
computing. Good references
are~\citet{yanofskymannucci:quantumcomputing,nielsenchuang:quantum}.
A quantum system is modelled by a finite-dimensional Hilbert space
\hilbspace. For example, \emph{qubits} are modelled by
$\mathbb{C}^2$. The category giving semantics to finite-dimensional
pure state quantum theory is therefore $\cat{FHilb}$, whose objects
are finite-dimensional Hilbert spaces, and whose morphisms are linear
maps. Categorical semantics for pure state quantum computing is the
groupoid $\cat{Unitary}$ of finite-dimensional Hilbert spaces as
objects with unitaries as morphisms. Both are rig categories under
direct sum $\oplus$ and tensor product $\otimes$. For example, a sum
type of a triple of qubits and a pair of qubits is modelled by
$(\mathbb{C}^2 \otimes \mathbb{C}^2 \otimes \mathbb{C}^2) \oplus
(\mathbb{C}^2 \otimes \mathbb{C}^2)$.

The pure \emph{states} of a quantum system modelled by a Hilbert space
\hilbspace\ are the vectors of unit norm, conventionally denoted by a \emph{ket}
$\ket{y} \in \hilbspace$. These are equivalently given by morphisms $\mathbb{C}
\to \hilbspace$ in $\cat{FHilb}$ that map $z \in \mathbb{C}$ to $z\ket{y} \in
\hilbspace$.  Dually, the functional $\hilbspace \to \mathbb{C}$ 
maps $x \in \hilbspace$ to the inner product $\braket{x|y}$ is
conventionally written as a \emph{bra} $\bra{x}$. Morphisms $\hilbspace \to
\mathbb{C}$ are also called \emph{effects}.

In fact, \cat{FHilb} is a dagger rig category. The \emph{dagger} of linear map $f \colon A \to B$ is uniquely determined via the inner product by $\braket{f(x)|y}=\braket{x|f^\dag(y)}$.
The dagger of a state is an effect, and vice versa. 
In quantum computing, pure states evolve along unitary gates. These are exactly the morphisms that are unitary in the sense of dagger categories in that $f^\dag \circ f = \id$ and $f \circ f^\dag = \id$, exhibiting the groupoid $\cat{Unitary}$ as a dagger subcategory of $\cat{FHilb}$.

Once orthonormal bases $\{\ket{i}\}$ and $\{\ket{j}\}$ for finite-dimensional Hilbert spaces \hilbspace\ and \hilbspacep\ are fixed, we can express morphisms $f \colon \hilbspace \to \hilbspacep$ as a matrix with entries $\braket{i|f|j}$. The dagger then becomes the complex conjugate transpose, the tensor product becomes the Kronecker product of matrices, and the direct sum becomes a block diagonal matrix. The Hilbert spaces $\mathbb{C}^n$ come with the canonical \emph{computational basis} consisting of the $n$ vectors with a single entry 1 and otherwise 0, also called the \emph{$Z$-basis} and denoted $\{\ket{0},\ldots,\ket{n-1}\}$.

In fact, there is a way to translate the category $\cat{FPInj}$ of
finite sets and partial injections to the category $\cat{FHilb}$, that
sends $\{0,\ldots,n-1\}$ to $\mathbb{C}^n$. This translation preserves
composition, identities, tensor product, direct sum, and dagger: it is
a dagger rig functor $\ell^2 \colon \cat{FPInj} \to \cat{FHilb}$, that
restricts to a dagger rig functor $\cat{FinBij} \to \cat{Unitary}$.
Thus reversible computing ($\cat{FinBij}$) is to classical reversible
theory ($\cat{FPInj}$) as quantum computing ($\cat{Unitary}$) is to
quantum theory ($\cat{FHilb}$).  In particular, in this way, the
Boolean controlled-controlled-not function (known as the \emph{Toffoli
  gate}), which is universal for reversible computing, transfers to a
quantum gate with the same name that acts on vectors.

There are many ways to embed classical reversible computing into
quantum computing like this, one for every uniform choice of
computational basis~\cite{heunen:ltwo}. If we only care about computation with qubits (rather than qutrits or the more general qudits), we could also send a bijection $f \colon
\{0,\ldots,2n-1\} \to \{0,\ldots,2n-1\}$ to $(H^{\otimes n})^\dag
\circ \ell^2(f) \circ H^{\otimes n}$, where $H$ is the Hadamard
matrix, to compute in the $X$ basis rather than the computational
($Z$) basis.

\subsection{Complementarity}
\label{subsec:complementarityagain}
\label{sec:complementarityagain}

A choice of basis $\{\ket{i}\}$ on an $n$-dimensional
Hilbert space \hilbspace\ defines a morphism $\delta \colon\hilbspace \to \hilbspace \otimes \hilbspace$ in $\cat{FHilb}$ that maps $\ket{i}$ to $\ket{ii} = \ket{i} \otimes \ket{i}$. In fact, the morphisms arising this way are characterised by certain equational laws that make them into a so-called \emph{classical structures}, or \emph{commutative special dagger Frobenius structures}~\cite[Chapter~5]{heunenvicary:cqt}. 
  \begin{align}\label{eq:classicalstructure1}
    (\delta \otimes \id_A) \circ \delta 
    &= (\id_A \otimes \delta) \circ \delta &
    \sigma_{A,A} \circ \delta &= \delta \\
    \label{eq:classicalstructure2}
    (\id_A \otimes \delta^\dag) \circ (\delta \otimes \id_A) 
    &= (\delta^\dag \otimes \id_A) \circ (\id_A \otimes \delta) &
    \delta^\dag \circ \delta &= \id_A 
  \end{align}
(To be precise, we deal with nonunital Frobenius structures, but in finite dimension there automatically exists a uniquely determined unit after all~\cite[Prop.~7]{abramskyheunen:hstar}.) In particular, because \hilbspace\ is finite-dimensional, the basis vectors determine a state $\sum_{i=1}^n \ket{i}$ that is in \emph{uniform superposition}. For the computational basis for qubits, this state is also denoted $\ket{+}=\frac{1}{\sqrt{2}}(\ket{0}+\ket{1})$. Similarly, we shorthand $\ket{-}=\frac{1}{\sqrt{2}}(\ket{0}-\ket{1})$. Now $\{\ket{+},\ket{-}\}$ forms an orthogonal basis for the qubit $\mathbb{C}^2$, called the \emph{$X$-basis}, that is different from the $Z$-basis. 

As far as picking a basis to treat as `the' computational basis is concerned, all bases are created equal. But once that arbitrary choice is fixed, some other bases are more equal than others. The $Z$-basis and the $X$-basis are \emph{mutually unbiased}, meaning that a state of the one basis and an effect of the other basis always give the same inner product: $\braket{0|+}=\braket{1|+}=\braket{0|-}=\braket{1|-}$. That is, measuring in one basis a state prepared in the other gives no information at all. This can also be expressed by an equation between the associated Frobenius structures $\delta_1,\delta_2\colon \hilbspace \to \hilbspace \otimes \hilbspace$ (see~\cite{coeckeduncan:zx} or~\cite[Chapter~6]{heunenvicary:cqt}):
  \begin{equation}\label{eq:complementarity}
    (\delta_1^\dag \otimes \id_A) \circ (\id_A \otimes \delta_2) \circ
    (\id_A \otimes \delta_2^\dag) \circ (\delta_1 \otimes \id_A) = \id_{A \otimes A} 
  \end{equation}
(To be precise, we adopt a simplified version using~\citet[Prop.~6.7]{heunenvicary:cqt}, and the fact that in finite dimension any injective morphism $\hilbspace \to \hilbspace$ is an isomorphism.)

Two complementary classical structures $\hilbspace \to \hilbspace
\otimes \hilbspace$ determine a unitary gate $\hilbspace
\to \hilbspace$, corresponding to the linear map that turns the basis
corresponding to one classical structure into the basis corresponding
to the other. In the case of the $Z$ and $X$ bases, this is the
Hadamard gate. Notice that the Hadamard gate is involutive: $H \circ
H=\id$.  It is well-known that there are three mutually unbiased bases
for qubits $\mathbb{C}^2$---namely the $Z$ and $X$ bases together with
another basis $Y$ we haven't discussed here---and there cannot be
four. Out of these three, only $Z$ and $X$ induce a unitary gate that
is involutive.

\subsection{Computational Universality}

The inner product of a Hilbert space \hilbspace\ lets us measure how
close two vectors $\ket{x},\ket{y} \in A$ are by looking at the norm
of their difference $\|x-y\|^2=\braket{x-y~|~x-y}$.  This leads to the
dagger rig category $\cat{Contraction}$ of finite-dimensional Hilbert
spaces and contractions: linear maps $f \colon A \to B$ satisfying
$\|f(a)\| \leq \|a\|$ for all $a \in A$. In $\cat{Contraction}$, the
notion of state is relaxed from a vector of unit length to a vector of
\emph{at most} unit length (these are sometimes called
\emph{subnormalised states} or simply \emph{substates}). This
categorical model adds to pure state quantum computation the ability
to terminate without a useful outcome, where the norm $\|x\|$ of a
state $x$ signifies the probability of a nondegenerate outcome when
measured; interpreting a state of norm $0$ as complete failure,
$\cat{Contraction}$ models pure state quantum computing where
computations may not terminate.

A finite set of unitary gates $\{U_1,\ldots,U_k\}$ on qubits is
\emph{strictly universal} when for any unitary $U$ and any
$\varepsilon>0$ there is a sequence of gates $U_{i_1} \circ \cdots
\circ U_{i_p}$ with distance at most $\varepsilon$ to $U$. This means
that any computation whatsoever can be approximated by a circuit from
the given set of gates up to arbitrary accuracy. The set is
\emph{computationally universal} when it can be used to simulate,
possibly using ancillas and/or encoding, to accuracy within
$\varepsilon>0$ any quantum circuit on $n$ qubits and $t$ gates from a
strictly universal set with only polylogarithmic overhead in $n$, $t$,
and $\tfrac{1}{\varepsilon}$. This means that the gate set can perform
general quantum computation without too much overhead.

\begin{theorem}\cite{aharonov:toffolihadamard,10.5555/2011508.2011515}\label{thm:aharonov}
  The Toffoli and Hadamard gate set is computationally universal. In fact, Toffoli is computationally universal in conjunction with any real basis-changing single-qubit unitary gate.
\end{theorem}

Notice that this theorem only needs sequential composition $\circ$ and parallel composition $\otimes$, and not sum types $\oplus$. Correspondingly, it only applies to Hilbert spaces of dimension $2^n$.

\section{The Classical Core: $\PiLang$}
\label{sec:pi}

\noindent Our eventual goal is to define a computationally universal
quantum programming language from two copies of a classical reversible
language. In this section, we review the syntax and semantics
of~$\PiLang$~\cite*{jamessabry:infeff}, a language that is universal
for reversible computing over finite types and whose semantics is
expressed in the rig groupoid of finite sets and bijections.

\begin{figure}[t]
\begin{align*}
  b & \defeq 0 \mid 1 \mid b+b \mid b \times b & \text{(value types)} \\
  t & \defeq b \fromto b & \text{(combinator types)} \\
  i & \defeq \pid \mid \swapp \mid \assocp \mid
  \associp \mid \unitepl \mid \unitipl & \text{(isomorphisms)} \\
    & \mid \swapt \mid \assoct \mid \associt 
      \mid \unitetl \mid \unititl \\
    & \mid \dist \mid \factor \mid \absorbl \mid \factorzr \\
  c & \defeq i \mid c \seqq c \mid c + c \mid c \times c \mid \inv~c & \text{(combinators)}
\end{align*}
\caption{\label{fig:pi}$\PiLang$ syntax}
\end{figure}

\begin{figure}[t]
\begin{equation*}
  \begin{array}{rcrclcl}
    \pid & \of & b &\fromto& b & \of & \pid \\
    \swapp & \of & b_1 + b_2 &\fromto& b_2 + b_1 & \of & \swapp \\
    \assocp & \of & (b_1 + b_2) + b_3 &\fromto& b_1 + (b_2 + b_3) & 
    \of & \associp \\
    \unitepl & \of & 0 + b &\fromto& b & \of & \unitipl \\
    \swapt & \of & b_1 \times b_2 &\fromto& b_2 \times b_1 & \of &
    \swapt \\
    \assoct & \of & (b_1 \times b_2) \times b_3 &\fromto& b_1 \times
    (b_2 \times b_3) & \of & \associt \\
    \unitetl & \of & 1 \times b &\fromto& b & \of & \unititl \\
    \dist & \of & (b_1 + b_2) \times b_3 &\fromto& (b_1 \times b_3)
    + (b_2 \times b_3) & \of & \factor \\
    \absorbl & \of & b \times 0 &\fromto& 0 & \of & \factorzr
  \end{array}
\end{equation*}
\begin{equation*}
  \frac{c_1 \of b_1 \fromto b_2 \quad c_2 \of b_2 \fromto b_3}{c_1 \seqq c_2 
    \of b_1 \fromto b_3}
  \qquad\qquad\qquad
  \frac{c \of b_1 \fromto b_2}{\inv~c \of b_2 \fromto b_1}
\end{equation*}
\begin{equation*}
  \frac{c_1 \of b_1 \fromto b_3 \quad c_2 \of b_2 \fromto b_4}{c_1 + c_2 
  \of b_1 + b_2 \fromto b_3 + b_4} \qquad
  \frac{c_1 \of b_1 \fromto b_3 \quad c_2 \of b_2 \fromto b_4}{c_1 \times c_2 
  \of b_1 \times b_2 \fromto b_3 \times b_4}
\end{equation*}
\caption{\label{fig:pitypes}Types for $\PiLang$ combinators}
\end{figure}

\begin{figure}[t]
\begin{align*}
  \ctrlgate~c &= \dist \seqq \pid + (\pid \times c ) \seqq \factor \\
  \xgate &= \swapp \\
  \cxgate &= \ctrlgate~\swapp \\
  \ccxgate &= \ctrlgate~\cxgate
\end{align*}
\caption{\label{fig:pid}Derived $\PiLang$ constructs.}
\end{figure}

\subsection{Syntax and Types}
\label{sub:ctrl}

In reversible boolean circuits, the number of input bits matches the
number of output bits. Thus, a key insight for a programming language
of reversible circuits is to ensure that each primitive operation
preserves the number of bits, which is just a natural number. The
algebraic structure of natural numbers as the free commutative
semiring (or, commutative rig), with $(0,+)$ for addition, and
$(1,\times)$ for multiplication then provides sequential, vertical,
and horizontal circuit composition. Generalizing these ideas, a typed
programming language for reversible computing should ensure that every
primitive expresses an isomorphism of finite types, \textit{i.e.}, a
permutation. The syntax of the language~$\PiLang$ in Fig.~\ref{fig:pi}
captures this concept.  Type expressions $b$ are built from the empty
type (0), the unit type (1), the sum type ($+$), and the product type
($\times$). A type isomorphism $c : b_1 \fromto b_2$ models a
reversible circuit that permutes the values in $b_1$ and $b_2$. These
type isomorphisms are built from the primitive identities $i$ and
their compositions. These isomorphisms are not ad hoc: they correspond
exactly to the laws of a \emph{rig} operationalised into invertible
transformations~\cite{10.1007/978-3-662-49498-1,CARETTE202215} which
have the types in Fig.~\ref{fig:pitypes}. Each line in the top part of
the figure has the pattern $c_1 : b_1 \fromto b_2 : c_2$
where $c_1$ and $c_2$ are self-duals; $c_1$ has type $b_1
\fromto b_2$ and $c_2$ has type $b_2 \fromto b_1$.

\medskip To see how this language expresses reversible circuits, we
first define types that describe sequences of booleans
($\bool^{n}$). We use the type $\bool = 1 + 1$ to represent booleans
with the left  injection representing false and the right injection
representing true. Boolean negation (the \xgate-gate) is
straightforward to define using the primitive combinator $\swapp$. We
can represent $n$-bit words using an n-ary product of boolean
values. To express the \cxgate- and \ccxgate-gates we need to encode a
notion of conditional expression. Such conditionals turn out to be
expressible using the distributivity and factoring identities of rigs
as shown in Fig.~\ref{fig:pid}. An input value of type $\bool \times
b$ is processed by the $\dist$ operator, which converts it into a
value of type $(1 \times b) + (1 \times b)$. Only in the right branch,
which corresponds to the case when the boolean is true, is the
combinator~\ensuremath{c} applied to the value of type~\ensuremath{b}.
The inverse of $\dist$, namely $\factor$ is applied to get the final
result. Using this conditional, \cxgate\ is defined as
$\ctrlgate~\xgate$ and the Toffoli \ccxgate\ is defined as
$\ctrlgate~\cxgate$. Because $\PiLang$ can express the Toffoli gate
and can generate ancilla values of type $1$ as needed, it is universal
for classical reversible circuits using the original 
construction of~\citet{10.1007/3-540-10003-2104}.

\begin{theorem}[$\PiLang$ Expressivity]
  $\PiLang$ is universal for classical reversible circuits, \textit{i.e.}, boolean
  bijections $\bool^n \to \bool^n$ (for any natural number $n$).
\end{theorem}

\begin{figure}
  \begin{equation*}
    \begin{array}{rclrclrcl}
      \multicolumn{9}{l}{\text{\textbf{Types}}} \\
      \sem{0} &=& O &&&& \sem{1} &=& I \\
      \sem{b_1 + b_2} &=& \sem{b_1} \oplus \sem{b_2} &&&&
      \sem{b_1 \times b_2} &=& \sem{b_1} \otimes \sem{b_2} \\ \\
      \multicolumn{9}{l}{\text{\textbf{Terms}}} \\

      \sem{\pid} &=& \id & 
      \sem{\inv~c} &=& \sem{c}^\dagger \\

      \sem{\swapp} &=& \sigma_\oplus &
      \sem{\assocp} &=& \alpha_\oplus &
      \sem{\associp} &=& \alpha_\oplus^{-1} \\

      \sem{\unitipl} &=& \rho_\otimes &
      \sem{\unitepl} & = & \lambda_{\oplus} \\

      \sem{\swapt} &=& \sigma_\otimes &
      \sem{\assoct} &=& \alpha_\otimes & 
      \sem{\associt} &=& \alpha_\otimes^{-1} \\

      \sem{\unititl} &=& \lambda_{\otimes}^{-1} &
      \sem{\unitetl} &=& \lambda_{\otimes} \\

      \sem{\dist} &=& \delta_R &
      \sem{\factor} &=& \delta_R^{-1} \\
 
      \sem{\absorbl} &=& \delta_0 & 
      \sem{\factorzr} &=& \delta_0^{-1} \\

      \sem{c_1 + c_2} &=& \sem{c_1} \oplus \sem{c_2} &
      \sem{c_1 \seqq c_2} &=& \sem{c_2} \circ \sem{c_1} &
      \sem{c_1 \times c_2} &=& \sem{c_1} \otimes \sem{c_2}
    \end{array}
  \end{equation*}
  \caption{The semantics of $\PiLang$ in rig groupoids with monoidal
    structures $(O,\oplus)$ and $(I,\otimes)$.}
  \label{fig:semantics}
\end{figure}

\subsection{Semantics}

By design, $\PiLang$ has a natural model in \emph{rig
  groupoids}~\cite{10.1007/978-3-662-49498-1,10.1145/3498667}. Indeed,
every atomic isomorphism of $\PiLang$ corresponds to a coherence
isomorphism in a rig category, while sequencing corresponds to
composition, and the two parallel compositions are handled by the two
monoidal structures. Inversion corresponds to the canonical dagger
structure of groupoids. This interpretation is summarised in
Fig.~\ref{fig:semantics}.

\section{Models of $\PiLang$ from Automorphisms}
\label{sec:automorphisms}

When $\PiLang$ is used as a stand-alone classical language, the rig
groupoid of finite sets and bijections is the canonical choice for the
semantics. In our case, as we aim to recover quantum computation from
two copies of $\PiLang$, the canonical choices need more structure
individually and, more importantly, must eventually be related to each
other in a particular way to express complementarity. We begin by
explaining the categorical construction needed to embed a rig groupoid
in the category $\Aut_a(\Unitary)$ parameterised by a family of
automorphisms $a$. We then use this construction to give two models
for $\PiLang$ embedded in two instances of $\Aut_a(\Unitary)$.

\subsection{$\Aut_a(\Unitary)$}
\label{sub:autunitary}

The semantics of the two copies of $\PiLang$ (which we call $\PizLang$
and $\PihLang$) will use a generalisation of the category $\Unitary$ to a
family of categories parameterised by automorphisms that are pre- and
post-composed with every morphism.

\begin{definition}
  Let $\cat{C}$ be a category, and for each object $C$ let $a_C \colon C \to C$ be an automorphism (that is not necessarily natural in any sense). Form a new category $\Aut_a(\cat{C})$ with:
  \begin{itemize}
    \item \textbf{Objects:} objects of $\cat{C}$.
    \item \textbf{Morphisms:} morphisms are those of the form $a_B^{-1}\circ f \circ a_A $ for every $f \colon A \to B$ of $\cat{C}$.
    \item \textbf{Composition:} as in $\cat{C}$.
  \end{itemize}
\end{definition}

\begin{proposition}\label{prop:Autriggroupoid}
  When $\cat{C}$ is a rig groupoid, so is $\Aut_a(\cat{C})$.
\end{proposition}
\begin{proof}
  To see that $\Aut_a(\cat{C})$ is a category, observe that, since $\Aut_a(\cat{C})$ inherits composition from~$\cat{C}$,  identities are just those from
  $\cat{C}$ since $a_A^{-1} \circ \id_A \circ a_A = a_A^{-1} \circ
  a_A = \id_A$, and composition is simply conjugated composition of
  morphisms from $\cat{C}$ since $a_C^{-1} \circ g \circ a_B \circ
  a_B^{-1} \circ f \circ a_A = a_C^{-1} \circ g \circ f \circ a_A$.
  Associativity and unitality of composition in $\Aut_a(\cat{C})$ 
  follows directly. That 
  $\Aut_a(\cat{C})$ is a groupoid when $\cat{C}$ is follows since for every 
  isomorphism $f$:
  \[
    a_B^{-1} \circ f \circ a_A \circ a_A^{-1} \circ f^{-1} \circ a_B
    = a_B^{-1} \circ f \circ f^{-1} \circ a_B
    = a_B^{-1} \circ a_B = \id_B
  \]
  and analogously $a_A^{-1} \circ f^{-1} \circ a_B \circ 
  a_B^{-1} \circ f \circ a_A = \id_A$.
  
  Supposing now $\cat{C}$ is symmetric monoidal, define a symmetric
  monoidal structure on objects as in $\cat{C}$ (with
  unit $I$ as in $\cat{C}$), and on morphisms $a_B^{-1} \circ f \circ a_A$ and $a_B'^{-1} \circ f' \circ a_A'$ by
  \[
    a_{B \otimes B'}^{-1} \circ (a_B \circ (a_B^{-1} \circ f \circ
    a_A) \circ a_A^{-1}) \otimes (a_B' \circ (a_B'^{-1} \circ f'
    \circ a_A') \circ a_A'^{-1}) \circ a_{A \otimes A'}
  \]
  in $\cat{C}$, which simplifies to
  $
    a_{B \otimes B'}^{-1} \circ (f \otimes f') \circ a_{A \otimes A'}
  $.
  In other words, monoidal products of morphisms in $\Aut_a(\cat{C})$ are merely 
  monoidal products of morphisms from $\cat{C}$ conjugated by the appropriate
  automorphisms. Coherence isomorphisms are those from $\cat{C}$, but 
  conjugated by the appropriate automorphisms. Bifunctoriality then follows because
  \[
    a_{A \otimes B}^{-1} \circ (\id_A \otimes \id_B) \circ a_{A \otimes B} = a_{A \otimes B}^{-1} \circ a_{A \otimes B} = \id_{A \otimes B}
  \]
  and for $f \colon B \to C$ and $g \colon C \to D$
  \begin{align*}
    & a_{A \otimes D}^{-1} \circ (\id_A \otimes g) \circ a_{A \otimes C} \circ a_{A \otimes C}^{-1} \circ (\id_A \otimes f) \circ a_{A \otimes B} \\
    & \quad = 
    a_{A \otimes D}^{-1} \circ (\id_A \otimes g) \circ (\id_A \otimes f)
    \circ a_{A \otimes B} \\
    & \quad = 
    a_{A \otimes D}^{-1} \circ (\id_A \otimes (g \circ f))
    \circ a_{A \otimes B}
  \end{align*}
  and similarly on the left, and finally for $f \colon A \to B$ and $g \colon C \to D$
  \begin{align*}
    & a_{D \otimes B}^{-1} \circ (g \otimes \id_B) \circ a_{C \otimes B} \circ a_{C \otimes B}^{-1} \circ (\id_C \otimes f) \circ a_{C \otimes A} \\
    & \quad = a_{D \otimes B}^{-1} \circ (g \otimes \id_B) \circ a_{C \otimes B} \circ a_{C \otimes B}^{-1} \circ (\id_C \otimes f) \circ a_{C \otimes A} \\
    & \quad = a_{D \otimes B}^{-1} \circ (g \otimes \id_B) \circ (\id_C \otimes f) \circ a_{C \otimes A} \\
    & \quad = a_{D \otimes B}^{-1} \circ (\id_D \otimes f) \circ (g \otimes \id_A) \circ a_{C \otimes A} \\
    & \quad = a_{D \otimes B}^{-1} \circ (\id_D \otimes f) \circ a_{D \otimes A} \circ a_{D \otimes A}^{-1} \circ (g \otimes \id_A) \circ a_{C \otimes A} \text.
  \end{align*}
  Naturality of coherence isomorphisms follows likewise since, for example,
  \begin{align*}
    & a_{A' \otimes (B' \otimes C')}^{-1} \circ \alpha \circ a_{(A' \otimes B') \otimes C'} \circ a_{(A' \otimes B') \otimes C'}^{-1} \circ (f \otimes g) \otimes h \circ a_{(A \otimes B) \otimes C} \\
    & \quad = a_{A' \otimes (B' \otimes C')}^{-1} \circ \alpha \circ (f \otimes g) \otimes h \circ a_{(A \otimes B) \otimes C} \\
    & \quad = a_{A' \otimes (B' \otimes C')}^{-1} \circ f \otimes
    (g \otimes h) \circ \alpha \circ a_{(A \otimes B) \otimes C} \\
    & \quad = a_{A' \otimes (B' \otimes C')}^{-1} \circ f \otimes
    (g \otimes h) \circ a_{A \otimes (B \otimes C)} \circ a_{A \otimes (B \otimes C)}^{-1} \circ \alpha \circ a_{(A \otimes B) \otimes C} 
  \end{align*}
  and similarly for the unitors and the symmetry. As for coherence conditions, 
  since these solely involve monoidal products of coherence isomorphisms and 
  identities, they transfer directly from~$\cat{C}$ by definition of 
  composition and monoidal products in $\Aut_a(\cat{C})$. For example, the 
  coherence condition that $\lambda \circ \sigma = \rho$ follows by
  \[
    a_{A}^{-1} \circ \lambda \circ a_{I \otimes A} \circ 
    a_{I \otimes A}^{-1} \circ \sigma \circ a_{A \otimes I} =
    a_{A}^{-1} \circ \lambda \circ \sigma \circ a_{A \otimes I} =
    a_{A}^{-1} \circ \rho \circ a_{A \otimes I}
  \]
  using the fact that the very same coherence condition holds in $\cat{C}$,
  and likewise for the other conditions. Thus $\Aut_a(\cat{C})$ is symmetric 
  monoidal when $\cat{C}$ is.
  
  Finally supposing that $\cat{C}$ is a rig category, it is specifically
  symmetric monoidal in two different ways, and so is $\Aut_a(\cat{C})$ by the
  argument above. To see that this extends to a rig structure, the remaining
  coherence isomorphisms $\delta_L$, $\delta_R$, and $\delta_0^L$, and 
  $\delta_0^R$ are defined in $\Aut_a(\cat{C})$ by conjugating those from 
  $\cat{C}$ with the appropriate automorphisms. That this satisfies the 
  coherence conditions of rig categories follows exactly as in the symmetric 
  monoidal case.
\end{proof}

\subsection{Models of $\PizLang$ and $\PihLang$: $\Unitary$ and $\Aut_{R\phid}(\Unitary)$}
\label{sub:models-pi}

The choice of semantics for $\PizLang$ is easy to justify: it will use
the trivial family of identity automorphisms, i.e, it will use the
canonical category $\Unitary$ itself. The semantics for $\PihLang$
will be ``rotated'' by some angle with respect to that of
$\PizLang$. By that, we mean that the semantics of $\PihLang$ will use
a family of automorphisms that is parameterised by a rotation matrix
$\mathbf{r\phid} = (
      \begin{smallmatrix}
        \cos{\phi} & -\sin{\phi} \\
        \sin{\phi} & \cos{\phi}
      \end{smallmatrix})$
for some yet-to-be determined angle $\phi$.       

\begin{definition}\label{def:piz-sem}
  The canonical model of $\PizLang$ is the rig groupoid $\Unitary$ of 
  finite-dimensional Hilbert spaces and unitaries.
\end{definition}

We recall that as a rig category, $\Unitary$ is \emph{semi-simple} in
the sense that all of its objects are generated by the rig structure
(this is a direct consequence of the fact that each finite-dimensional
Hilbert space is isometrically isomorphic to $\mathbb{C}^n$ for some
$n$~\cite{loaiza:short}). In other words, every object in $\Unitary$
can be written (up to isomorphism) using the two units $O$ and $I$ as well as the monoidal
product $\otimes$ and sum $\oplus$. We will use this fact to define a
family of automorphisms $R\phid_A$ in $\Unitary$ which will be used to
form a model of $\PihLang$.

\begin{definition}\label{def:rx}
  Given an angle $\phi$, we define a family $R\phid_A$ of
  automorphisms in $\Unitary$ as follows:
  \begin{equation*}
    \begin{array}{rclrcl}
      R\phid_O & = & \id_O & 
      R\phid_I & = & \id_I \\
      R\phid_{A \otimes B} & = & R\phid_A \otimes R\phid_B \\
      R\phid_{A \oplus B} & = & \theta_A^{-1} \oplus \theta_B^{-1} \circ 
      \mathbf{r}\phid
      \circ \theta_A \oplus \theta_B
      & \multicolumn{3}{l}{\text{(when $A \simeq I \simeq B$)}} \\
      R\phid_{A \oplus B} & = & R\phid_A \oplus R\phid_B
      & \multicolumn{3}{l}{\text{(when $A \not\simeq I$ or $B \not\simeq I$)}}
    \end{array}
  \end{equation*}
\end{definition}
\noindent The morphisms $\theta_A$ and $\theta_B$ in this definition refer to the isomorphisms witnessing $A \simeq I$ and $B \simeq I$ respectively. In particular, this definition requires one to decide isomorphism with $I$. This sounds potentially difficult, but is fortunately very simple: an object is isomorphic to $I$ iff it can be turned into $I$ by eliminating additive units $O$ and multiplicative units $I$ using the unitors $\lambda_\oplus$, $\rho_\oplus$, $\lambda_\otimes$, and $\rho_\otimes$ as well as the associators $\alpha_\oplus$ and $\alpha_\otimes$ as necessary.

Essentially, $\Aut_{R\phid}(\Unitary)$ consists of unitaries in which
qubit (sub)systems are conjugated by the unitary $\mathbf{r}\phid$.
This is significant, because it means that the additive symmetry
$\sigma_\oplus$ on $I \oplus I$ is no longer $(\begin{smallmatrix} 0 &
  1 \\ 1 & 0
\end{smallmatrix})$ as usual, but instead the potentially much more interesting gate:
\[
\begin{pmatrix}
  \cos{\phi} & \sin{\phi} \\ -\sin{\phi} & \cos{\phi}
\end{pmatrix}
\begin{pmatrix}
  0 & 1 \\
  1 & 0
\end{pmatrix}
\begin{pmatrix}
  \cos{\phi} & -\sin{\phi} \\
  \sin{\phi} & \cos{\phi}
\end{pmatrix}
=
\begin{pmatrix}
  \sin{2\phi} & \cos{2\phi} \\
  \cos{2\phi} & -\sin{2\phi}
\end{pmatrix}\enspace.
\]
This leads us to the family of models of $\PihLang$.

\begin{definition}\label{def:pih-sem}
  Given a value for $\phi$, a model of $\PihLang$ is the
  rig groupoid $\Aut_{R\phid}(\Unitary)$ of finite-dimensional Hilbert spaces
  and unitaries of the form $(R\phid_B)^{-1} \circ U \circ R\phid_A$.
\end{definition}

\section{$\PizhLang$ from Amalgamation} 
\label{sec:amalgamation}

The aim of this section is to define the language $\PizhLang$ that
combines the separate definitions of $\PizLang$ and $\PihLang$ into a
combined language that interleaves expressions from each. We begin by
explaining the amalgamations of categories in
Secs.~\ref{sub:amalgamate}, \ref{sub:amalgamateSC},
and~\ref{sub:moreamalg}. We use these constructions to define
categorical models of $\PizhLang$ in Sec.~\ref{sub:pizhmodel}. These
models justify the definition of $\PizhLang$ as an arrow over the
individual sublanguages as shown in Sec.~\ref{sub:pizh}.

\subsection{Amalgamation of Categories}
\label{sub:amalgamate}

Programs in $\PizhLang$ are formal compositions of $\PizLang$ programs and $\PihLang$ programs, combined in a way that respects product types. To account for this semantically, we need a way to combine \emph{models} of $\PizLang$ and $\PihLang$ in a way that preserves the monoidal product. This construction is known as the \emph{amalgamation of categories} (see, \textit{e.g.}, ~\citet{macdonaldscull:amalg}). In this section, we recall this construction in the slightly simpler case where the two categories have the same objects, and go on to extend it to the symmetric monoidal case.

\begin{definition}\label{def-amalg}
  Given two categories~$\cat{C}$ and $\cat{D}$ with the same objects, form a
  new category $\Amalg(\cat{C}, \cat{D})$ as follows:
  \begin{itemize}
    \item \textbf{Objects}: Objects of $\cat{C}$ (equivalently $\cat{D}$).
    \item \textbf{Morphisms:} Morphisms $A_1 \to A_{n+1}$ are equivalence 
    classes of finite lists $[f_n, \dots, f_1]$ of morphisms $f_i \colon A_i \to A_{i+1}$ of $\cat{C}$ or $\cat{D}$ tagged with their category of origin,
    under the equivalence $\sim$ below.
    \item \textbf{Identities:} Empty lists $[]$.
    \item \textbf{Composition:} Concatenation of lists, $[g_n, \dots, g_1]
    \circ [f_m, \dots, f_1] = [g_n, \dots, g_1, f_m, \dots, f_1]$.
  \end{itemize}
  When the origin category is important we will write, \textit{e.g.}, $f^\cat{C}$ to mean
  that $f$ is tagged with $\cat{C}$ and so originated from this category. Let
  $\sim$ denote the least equivalence satisfying
  \begin{multicols}{2}
  \noindent
    \begin{equation}
      \label{eq:identity-equiv}
      [\id] \sim []
    \end{equation}
    \begin{equation}
      \label{eq:comp-equiv}
      [f^{\cat{A}}, g^{\cat{A}}] \sim [f^{\cat{A}} \circ g^{\cat{A}}]
    \end{equation}
  \end{multicols}
  \noindent as well as the congruence:
  \begin{equation}
    \frac{[f_n, \dots, f_1] \sim [f'_{n'}, \dots, f'_1] \qquad [g_{m}, \dots,
    g_1] \sim [g'_{m'}, \dots, g'_1]}{[f_n, \dots, f_1] \circ [g_m, \dots, g_1]
    \sim [f'_{n'}, \dots, f'_1] \circ [g'_{m'}, \dots, g'_1]} 
    \label{eq:comp-cong}
  \end{equation}
\end{definition}
Note that the inner composition in \eqref{eq:comp-equiv} refers to composition
in the category $\cat{A}$, which in turn refers to either $\cat{C}$ or
$\cat{D}$. To verify that this forms a category, we notice that concatenation
of lists is straightforwardly associative and has the empty list as unit;
however, since these are not lists per se but in fact equivalence classes of
lists, we must also check that composition is well-defined. To see this,
consider the normalisation procedure that repeats the following two steps until
a fixed point is reached:
\begin{enumerate}
  \item Remove all identities using \eqref{eq:identity-equiv} and 
  \eqref{eq:comp-cong}.
  \item Compose all composable adjacent morphisms using
  \eqref{eq:comp-equiv} and \eqref{eq:comp-cong}.
\end{enumerate}
That a fixed point is always reached follows by the fact that both of these
steps are monotonically decreasing in the length of the list, which is always
finite.

As we would hope, there are straightforward embeddings $\cat{C} \to
\Amalg(\cat{C},\cat{D}) \leftarrow \cat{D}$.
\begin{proposition}\label{prop-embed-amalg}
  There are embeddings $\mathscr{E}_L : \cat{C} \to 
  \Amalg(\cat{C},\cat{D})$
  and $\mathscr{E}_R : \cat{D} \to \Amalg(\cat{C},\cat{D})$ given on 
  objects by $X \mapsto X$ and on morphisms by $f \mapsto [f]$.
\end{proposition}
\begin{proof}
  $\mathscr{E}_L(\id) = [\id] \sim []$ and $\mathscr{E}_L(g \circ
  f) = [g \circ f] \sim [g, f] = \mathscr{E}_L(g) \circ
  \mathscr{E}_L(f)$, likewise for $\mathscr{E}_R$.
\end{proof}

Since $\PizLang$ and $\PihLang$ are both reversible, we would also expect $\PizhLang$ to be so, simply by taking inverses pointwise. We can formalise this intuition by showing that the amalgamation of \emph{groupoids} is, again, a groupoid.
\begin{proposition}\label{prop-amalg-groupoid}
  $\Amalg(\cat{C}, \cat{D})$ is a groupoid when $\cat{C}$ and $\cat{D}$ are.
\end{proposition}
\begin{proof}
  Define $[f_n, \dots, f_1]^{-1} = [f_1^{-1}, \dots, f_n^{-1}]$ (where
  $f_i^{-1}$ is the inverse to $f_i$ in the origin category), and proceed by
  induction on $n$. When $n=0$, $[] \circ []^{-1} = [] \circ [] = []$. Assuming 
  the inductive hypothesis on all lists of length $n$ we see on lists of length 
  $n+1$ that
  \begin{align*}
    [f_{n+1}, \dots, f_1] \circ [f_{n+1}, \dots, f_1]^{-1} 
    & = [f_{n+1}, \dots, f_1, f_1^{-1}, \dots, f_{n+1}^{-1}] \\
    & \sim [f_{n+1}, \dots, f_1 \circ f_1^{-1}, \dots, f_{n+1}^{-1}] \\
    & = [f_{n+1}, \dots, f_2, \id, f_2^{-1}, \dots, f_{n+1}^{-1}] \\
    & \sim [f_{n+1}, \dots, f_2, f_2^{-1}, \dots, f_{n+1}^{-1}] \\
    & = [f_{n+1}, \dots, f_2] \circ [f_2^{-1}, \dots, f_{n+1}^{-1}] \\
    & = [f_{n+1}, \dots, f_2] \circ [f_{n+1}, \dots, f_2]^{-1} \\
    &= []
  \end{align*}
  where the last identity follows by the inductive hypothesis.
\end{proof}

\subsection{Amalgamation of Symmetric Monoidal Categories}
\label{sub:amalgamateSC}

Categorically, the amalgamation of categories (with the same objects)
has a universal property as a pushout of (identity-on-objects)
embeddings in the category $\cat{Cat}$ of (small) categories and
functors between them~\cite{macdonaldscull:amalg}. While this
amalgamation is a good first step towards a model of $\PizhLang$, it
is not enough. This is because it is only a pushout of mere functors
between unstructured categories, so it will not necessarily respect
\emph{structure} present in the categories being amalgamated, such as
symmetric monoidal structure. In this section, we will extend the
amalgamation of categories to one for symmetric monoidal
categories. In the next section, we show that this yields an arrow
over the symmetric monoidal categories involved.

\begin{definition}\label{def:symmonamalg}
  Given two symmetric monoidal categories $\cat{C}$ and $\cat{D}$ with the same 
  objects, such that their symmetric monoidal products agree on objects 
  (specifically, their units are the same), form a new category 
  $\SymMonAmalg(\cat{C}, \cat{D})$ as follows:
  \begin{itemize}
    \item \textbf{Objects}, \textbf{Morphisms}, \textbf{Identities}, and 
    \textbf{Composition} as in $\Amalg(\cat{C}, \cat{D})$ 
    (Def.~\ref{def-amalg}).
    \item \textbf{Monoidal unit:} $I$, the monoidal unit of $\cat{C}$ and 
    $\cat{D}$.
    \item \textbf{Monoidal product:} On objects, define $A \otimes B$ to be as
    in $\cat{C}$ and $\cat{D}$. On morphisms, define $[f_n, \dots, f_1] \otimes
    [g_m, \dots, g_1] = [f_n \otimes \id, \dots, f_1 \otimes \id, \id \otimes
    g_m, \dots, \id \otimes g_1]$ where $f_i \otimes \id$ and $\id \otimes g_j$
    are formed in the origin category of $f_i$ and $g_j$ respectively, up to
    the extended equivalence below.
    \item \textbf{Coherence isomorphisms:} The coherence isomorphisms $\alpha$, 
    $\sigma$, $\lambda$, $\rho$, and their inverses are given by equivalence 
    classes of lifted coherence isomorphisms from $\cat{C}$ and $\cat{D}$ 
    (\textit{e.g.}, $[\alpha^\cat{C}]$) up to the extended equivalence below.
  \end{itemize}
  The extended equivalence is the least one containing 
  \eqref{eq:identity-equiv}, 
  \eqref{eq:comp-equiv}, and \eqref{eq:comp-cong} from 
  Def.~\ref{def-amalg} as well as
  \begin{align}
    [f \otimes \id, \id \otimes g] & \sim [\id \otimes g, f \otimes \id] 
    \label{eq-bifun}
    \\
    \label{eq-lambda}
      [\alpha^\cat{C}] & \sim [\alpha^\cat{D}] &
      [\sigma^\cat{C}] & \sim [\sigma^\cat{D}] &
      [\lambda^\cat{C}] & \sim [\lambda^\cat{D}] &
      [\rho^\cat{C}] & \sim [\rho^\cat{D}] 
  \end{align}
  in addition to the congruence:
  \begin{equation}\label{eq-cong}
    \frac{[f_n, \dots, f_1] \sim [f'_{n'}, \dots, f'_1] \qquad [g_m, \dots,
    g_1] \sim [g'_{m'}, \dots, g'_1]}{[f_n, \dots, f_1] \otimes [g_m, \dots,
    g_1] \sim [f'_{n'}, \dots, f'_1] \otimes [g'_{m'}, \dots, 
    g'_1]}
  \end{equation}
\end{definition}

Note that \eqref{eq-bifun} above holds even when $f$ and $g$ originate from
different categories, such that this is not simply a consequence of
\eqref{eq:comp-equiv} and bifunctoriality in the origin category.

It follows that this defines a category, but it remains to show that this actually defines a symmetric monoidal category.

\begin{proposition}\label{prop:amalgsymmon}
  $\SymMonAmalg(\cat{C},\cat{D})$ is symmetric monoidal.
\end{proposition}
\begin{proof}
  We must show that the monoidal product is bifunctorial, that coherence 
  isomorphisms are natural, and that coherence conditions are satisfied. For 
  bifunctoriality, we see first that the monoidal product is functorial in each 
  component since $[] \otimes [] = []$ and
  \begin{align*}
    ([] \otimes [f_n, \dots, f_1]) \circ ([] \otimes [g_m, \dots, g_1])
    & = [\id \otimes f_n, \dots, \id \otimes f_1] \circ [\id
    \otimes g_m, \dots, \id \otimes g_1] \\
    & = [\id \otimes f_n, \dots, \id \otimes f_1, \id
    \otimes g_m, \dots, \id \otimes g_1] \\
    & = [] \otimes [f_n, \dots, f_1, g_m, \dots, g_1] \\
    & = [] \otimes ([f_n, \dots, f_1] \circ [g_m, \dots, g_1])
  \end{align*}
  and similarly with the identity on the right. For bifunctoriality, we expand
  \begin{align*}
    ([f_n, \dots, f_1] \otimes []) \circ ([] \otimes [g_m, \dots, g_1]) 
    & = [f_n \otimes \id, \dots, f_1 \otimes \id] \circ [\id \otimes g_m,
    \dots, \id \otimes g_1] \\
    & = [f_n \otimes \id, \dots, f_1 \otimes \id, \id \otimes g_m, \dots, \id
    \otimes g_1]
  \end{align*}
  and inductively applying \eqref{eq-bifun} $m$ times for each $f_i$ shows
  \begin{align*}
    ([f_n, \dots, f_1] \otimes []) \circ ([] \otimes [g_m, \dots, g_1]) & = 
    [f_n \otimes \id, \dots, f_1 \otimes \id, \id \otimes g_m, \dots, \id
    \otimes g_1] \\
    & \sim [\id \otimes g_m, \dots, \id \otimes g_1, f_n \otimes \id, \dots, f_1
    \otimes \id] \\
    & = ([] \otimes [g_m, \dots, g_1]) \circ ([f_n, \dots, f_1] \otimes [])
  \end{align*}
  To see that coherence isomorphisms are natural, we consider the case of the
  left unitor $\lambda$: that this is natural amounts to showing
  $[\lambda^\cat{C}] \circ [I \otimes f_n, \dots, I \otimes f_1] = [f_n, \dots,
  f_1] \circ [\lambda^\cat{C}]$ (we could have equivalently used
  $[\lambda^\cat{D}]$, but it makes no difference by \eqref{eq-lambda}). Assume
  without loss of generality that $[f_n, \dots, f_1]$ alternates between
  morphisms from $\cat{C}$ and morphisms from $\cat{D}$, i.e., $f_n$,
  $f_{n-2}$, $f_{n-4}$, etc. originate in $\cat{C}$, while $f_{n-1}$,
  $f_{n-3}$, $f_{n-5}$, etc. originate in $\cat{D}$. Then
  \begin{align*}
    [\lambda^\cat{C}] \circ [I \otimes f_n, I \otimes f_{n-1}, \dots, I \otimes
    f_1] & = [\lambda^\cat{C}, I \otimes f_n, I \otimes f_{n-1}, \dots, I
    \otimes f_1] \\ 
    & \sim [\lambda^\cat{C} \circ I \otimes f_n, I \otimes f_{n-1}, \dots, I
    \otimes f_1] \\
    & = [f_n \circ \lambda^\cat{C}, I \otimes f_{n-1}, \dots, I
    \otimes f_1] \\
    & \sim [f_n, \lambda^\cat{C}, I \otimes f_{n-1}, \dots, I
    \otimes f_1] \\
    & \sim [f_n, \lambda^\cat{D}, I \otimes f_{n-1}, \dots, I
    \otimes f_1] \\
    & \sim [f_n, \lambda^\cat{D} \circ I \otimes f_{n-1}, \dots, I
    \otimes f_1] \\
    & = [f_n, f_{n-1} \circ \lambda^\cat{D}, \dots, I
    \otimes f_1] \\
    & \sim [f_n, f_{n-1}, \lambda^\cat{D}, \dots, I
    \otimes f_1]
  \end{align*}
  and so inductively using \eqref{eq:comp-equiv} to compose the coherence
  isomorphism with the next morphism in line, naturality of that coherence
  isomorphism in the origin category, and the ability to swap coherence
  isomorphisms between the two categories using \eqref{eq-lambda} shows that
  this is natural. The cases for the remaining coherence isomorphisms are
  entirely analogous.
  
  For the coherence conditions, the strategy is similar. Since every coherence 
  condition of symmetric monoidal categories concerns compositions of monoidal 
  products of coherence isomorphisms and identities, we can use the 
  equivalences between coherence isomorphisms \eqref{eq-lambda} 
  as well as the two congruences \eqref{eq:comp-cong} and \eqref{eq-cong} to 
  bring every morphism involved into the same category, e.g., $\cat{C}$. Then, 
  we can inductively use the composition rule \eqref{eq:comp-equiv} to turn the 
  list into a singleton corresponding to the exact same coherence condition,
  but now phrased in $\cat{C}$. We can then finally use coherence in $\cat{C}$ 
  to show the desired identity, and then run the entire process in reverse to
  lift this in $\SymMonAmalg(\cat{C}, \cat{D})$. For example, that the triangle 
  \[\begin{tikzcd}[ampersand replacement=\&]
    {(A \otimes I) \otimes B} \&\& {A \otimes (I \otimes B)} \\
    \& {A \otimes B}
    \arrow["\alpha", from=1-1, to=1-3]
    \arrow["{\id \otimes \lambda}", from=1-3, to=2-2]
    \arrow["{\rho \otimes \id}"', from=1-1, to=2-2]
  \end{tikzcd}\]
  commutes (this is one of the coherence conditions) follows by
  \begin{equation*}
    ([] \otimes [\lambda^\cat{D}]) \circ [\alpha^\cat{C}] \sim ([] \otimes 
    [\lambda^\cat{C}]) \circ [\alpha^\cat{C}] =
    [\id \otimes \lambda^\cat{C}, \alpha^\cat{C}] \sim
    [\id \otimes \lambda^\cat{C} \circ \alpha^\cat{C}] = [\rho^\cat{C} \otimes 
    \id] = [\rho^\cat{C}] \otimes [] \enspace.
  \end{equation*}
  Again, the choice of $[\lambda^\cat{D}]$, $[\alpha^\cat{C}]$, 
  $[\rho^\cat{C}]$ is purely for illustration, and we could equivalently have 
  chosen ones from the other category by the coherence equivalences. All other 
  coherence conditions may be shown using the same strategy.
\end{proof}

As with the amalgamation of mere categories, one can show that this extends to a pushout of monoidal embeddings. Further, the embeddings we presented earlier into the amalgamation of mere categories extends to well-behaved ones in the symmetric monoidal case well:

\begin{proposition}\label{prop:monembamalg}
  There are strict monoidal embeddings $\mathscr{E}_L : \cat{C} \to 
  \SymMonAmalg(\cat{C},\cat{D})$ and $\mathscr{E}_R : \cat{D} \to 
  \SymMonAmalg(\cat{C},\cat{D})$ given by $X \mapsto X$ on objects and $f 
  \mapsto [f]$ on morphisms.
\end{proposition}
\begin{proof}
  We show the case for $\mathscr{E}_L$, as $\mathscr{E}_R$ is 
  entirely analogous. $\mathscr{E}_L$ was shown to be functorial in 
  Prop.~\ref{prop-embed-amalg}, so suffices to show that it preserves 
  coherence isomorphisms and the monoidal product exactly on objects and
  morphisms. On objects $\mathscr{E}_L(A \otimes B) = A \otimes B = 
  \mathscr{E}_L(A) \otimes \mathscr{E}_L(B)$. On morphisms 
  $
  \mathscr{E}_L(f \otimes g) = [f \otimes g] = [f \otimes \id \circ \id
  \otimes g] \sim [f \otimes \id, \id
  \otimes g] = [f] \otimes [g] = \mathscr{E}_L(f) \otimes
  \mathscr{E}_L(g)
  $. Finally, on coherence isomorphisms $\beta$, $\mathscr{E}_L(\beta) 
  = [\beta^\cat{C}] \sim [\beta^\cat{D}] = \mathscr{E}_L(\beta)$, as 
  desired.
\end{proof}

It follows straightforwardly using the construction of inverses in Prop.~\ref{prop-amalg-groupoid} that amalgamation preserves symmetric monoidal \emph{groupoids} as well:

\begin{corollary}
  $\SymMonAmalg(\cat{C}, \cat{D})$ is a symmetric monoidal groupoid when 
  $\cat{C}$ and $\cat{D}$ are.
\end{corollary}

We conclude this section with a lemma stating that whenever a
functor \emph{out of} a symmetric monoidal amalgamation is needed, it
is sufficient to consider functors out of each of the underlying
categories. The lemma will be used to prove the existence of a
computationally universal model of $\PizhLang$ in
Theorem.~\ref{thm:approx-universality-pizh}.

\begin{lemma}\label{lem:monamalgfree}
  Let $\cat{C}$ and $\cat{D}$ be symmetric monoidal categories with the same
  objects such that their monoidal structures agree. For any other symmetric
  monoidal category $\cat{E}$, to give a strict monoidal identity-on-objects
  functor $\SymMonAmalg(\cat{C}, \cat{D}) \to \cat{E}$ is to give strict
  monoidal identity-on-objects functors $\cat{C} \to \cat{E}$ and $\cat{D} \to
  \cat{E}$.
\end{lemma}
\begin{proof}
  Given a strict monoidal identity-on-objects functor $F :
  \SymMonAmalg(\cat{C}, \cat{D}) \to \cat{E}$, we compose with the
  (strict monoidal identity-on-objects) functors $\mathscr{E}_L$ and
  $\mathscr{E}_R$ to obtain the required functors $F \circ \mathscr{E}_L \colon \cat{C} \to \cat{E}$ and $F \circ \mathscr{E}_R \colon \cat{D} \to \cat{E}$.
  
  In the other direction, given strict monoidal identity-on-objects functors $G \colon \cat{C} \to \cat{E}$ and $H \colon \cat{D} \to \cat{E}$, we define a functor
  $F_{G,H} \colon \SymMonAmalg(\cat{C}, \cat{D}) \to \cat{E}$ on objects by
  $F_{G,H}(A) = G(A) = H(A) = A$. Given some morphism $[f_n, \dots, f_1]$,
  assume without loss of generality that each $f_i$ originates in $\cat{C}$ for
  all even $i$, and in $\cat{D}$ for all odd $n$. We define $F_{G,H}([]) = \id$ 
  and
  $$
    F_{G,H}([f_n, \dots, f_1]) = G(f_n) \circ H(f_{n-1}) \circ \cdots \circ 
    H(f_1)
  $$
  This is immediately functorial. To see that it is strict monoidal, $F_{G,H}(A\otimes B) = F_{G,H}(A) \otimes F_{G,H}(B)$ follows trivially, while
  \begin{align*}
    F_{G,H}([f_n, \dots, f_1] \otimes [g_m, \dots, g_1]) & = 
    F_{G,H}([f_n \otimes \id, \dots, f_1 \otimes \id, \id \otimes g_m, \dots, 
    \id \otimes g_1]) \\
    & = G(f_n \otimes \id) \circ \cdots \circ H(f_1 \otimes \id) \circ
    G(\id \otimes g_m) \circ \cdots \circ H(\id \otimes g_1) \\
    & = G(f_n) \otimes \id \circ \cdots \circ H(f_1) \otimes \id \circ
    \id \otimes G(g_m) \circ \cdots \circ \id \otimes H(g_1) \\
    & = (G(f_n) \circ \cdots \circ H(f_1)) \otimes \id \circ
    \id \otimes (G(g_m) \circ \cdots \circ H(g_1)) \\
    & = (G(f_n) \circ \cdots \circ H(f_1)) \otimes (G(g_m) \circ \cdots \circ 
    H(g_1)) \\
    & = F_{G,H}([f_n, \dots, f_1]) \otimes F_{G,H}([g_m, \dots, g_1])
  \end{align*}
  That this preserves coherence isomorphisms such as the associator 
  $[\alpha^\cat{C}] \sim [\alpha^\cat{D}]$ 
  follows by $G(\alpha) = F_{G,H}([\alpha^\cat{C}]) = F_{G,H}([\alpha^\cat{D}]) 
  = H(\alpha) = \alpha$, and similarly for the unitors $\lambda$, $\rho$ and 
  symmetry $\sigma$. Further, $F_{G,H}$ is clearly uniquely determined by $G$ 
  and $H$, i.e., $F_{G,H} \circ \mathscr{E}_L = G$ and $F_{G,H} \circ 
  \mathscr{E}_R = H$.
\end{proof}

\subsection{The Amalgamation Arrow}
\label{sub:moreamalg}

Semantically, arrows correspond to (identity-on-objects) strict
\emph{premonoidal} functors between \emph{premonoidal}
categories~\cite{powerrobinson:premonoidal,jacobsheunenhasuo:arrows},
a special case of these being the more well-behaved
(identity-on-objects) strict \emph{monoidal} functors between
\emph{monoidal} categories. In this way, the strict monoidal functors
$\Unitary \to \SymMonAmalg(\Unitary, \Aut_{R\phid}(\Unitary))
\leftarrow \Aut_{R\phid}(\Unitary)$ of
Prop.~\ref{prop:monembamalg} provide a semantics for the arrow
combinators. We summarise this is in the following proposition.

\begin{proposition}
  The strict monoidal functors $\mathscr{E}_L$ and $\mathscr{E}_R$ are 
  arrows over the categories $\Unitary$ and $\Aut_{R\phid}(\Unitary)$.
\end{proposition}

\subsection{Model of $\PizhLang$: $\SymMonAmalg(\Unitary, \Aut_{R\phid}(\Unitary))$}
\label{sub:pizhmodel}

Given models of $\PizLang$ (Def.~\ref{def:piz-sem}) and $\PihLang$ (Def.~\ref{def:pih-sem}), using Def.~\ref{def:symmonamalg} we can give a model of $\PizhLang$:

\begin{definition}\label{def:pizh-sem}
  Given a value for $\phi$, a model of $\PizhLang$ is the symmetric
  monoidal groupoid
  \[\SymMonAmalg(\Unitary, \Aut_{R\phid}(\Unitary))\]
  with $\Unitary$ and $\Aut_{R\phid}(\Unitary)$ considered as
  symmetric monoidal groupoids equipped with their monoidal products
  $(\otimes, I)$.
\end{definition}

In other words, $\SymMonAmalg(\Unitary, \Aut_{R\phid}(\Unitary))$ identifies the monoidal products $(\otimes, I)$ in $\Unitary$ and $\Aut_{R\phid}(\Unitary)$, but leaves their respective monoidal sums $(\oplus, O)$ alone. This may seem like a very curious choice---perhaps even a wrong one!---but is done for very deliberate reasons, which we describe here. 

First, identifying the two monoidal products is entirely reasonable, since $R\phid_{A \otimes B} = R\phid_A \otimes R\phid_B$, so the monoidal product on morphisms in $\Aut_{R\phid}(\Unitary)$ is really 
\[\begin{array}{rcl}
(R\phid_{A' \otimes B'})^{-1} \circ (f \otimes g) \circ R\phid_{A \otimes B} &=& (R\phid_{A'})^{-1} \otimes (R\phid_{B'})^{-1} \circ (f \otimes g) \circ R\phid_{A} \otimes R\phid_{B} \\
&= &((R\phid_{A'})^{-1} \circ f \circ R\phid_{A}) \otimes ((R\phid_{B'})^{-1} \circ g \otimes R\phid_{B})
\end{array}\]
\textit{i.e.}, the monoidal product in $\Unitary$ of morphisms from $\Aut_{R\phid}(\Unitary)$ (on objects, the two monoidal products agree on the nose). From this it also follows that the coherence isomorphisms for the monoidal product (\textit{i.e.}, the associator $\alpha_\otimes$, unitors $\lambda_\otimes$ and $\rho_\otimes$, and symmetry $\sigma_\otimes$) in $\Aut_{R\phid}(\Unitary)$ all coincide with those in $\Unitary$ by naturality, since, \textit{e.g.}
\[\begin{array}{rcl}
  (R\phid_{A \otimes (B \otimes C)})^{-1} \circ \alpha \circ 
  R\phid_{(A \otimes B) \otimes C}
  &=& (R\phid_{A})^{-1} \otimes ((R\phid_B)^{-1} \otimes (R\phid_C)^{-1}) \circ \alpha \circ 
  (R\phid_{A} \otimes R\phid_B) \otimes R\phid_C \\
  &=& (R\phid_{A})^{-1} \otimes ((R\phid_B)^{-1} \otimes (R\phid_C)^{-1}) \circ 
  R\phid_{A} \otimes (R\phid_B \otimes R\phid_C) \circ \alpha \\
  &=& \alpha
\end{array}\]
and likewise for the unitors and symmetry. As such, all of the morphisms that are identified by the symmetric monoidal amalgamation $\SymMonAmalg(\Unitary, \Aut_{R\phid}(\Unitary))$ are ones which were equal to begin with.

Second, one may wonder why we do not go further and identify the monoidal \emph{sums} in $\Unitary$ and $\Aut_{R\phid}(\Unitary)$ as well. In short, this is because it would confine $\PizhLang$ to being a classical language! We saw in Sec.~\ref{sub:models-pi} that the symmetry of the monoidal sum $\sigma_\oplus$ in $\Aut_{R\phid}(\Unitary)$ was
$(\begin{smallmatrix}
  \sin{2\phi} & \cos{2\phi} \\
  \cos{2\phi} & - \sin{2\phi}
\end{smallmatrix})$
whereas in $\Unitary$ it is the usual swap
$(\begin{smallmatrix}
  0 & 1 \\ 1 & 0 
\end{smallmatrix})$, and, indeed, the fact that we have both of these is central to our approach. However, identifying the monoidal sums would force us to identify these as well, destroying any hope of $\PizhLang$ being more expressive than $\PizLang$ or $\PihLang$ on their own. One cannot even hope to identify the monoidal sums as mere monoidal structures (as opposed to as \emph{symmetric} monoidal structures), since the bifunctoriality clause of the equivalence (\textit{i.e.}, clause \eqref{eq-bifun} of Def.~\ref{def:symmonamalg}) fails for $\Unitary$ and $\Aut_{R\phid}(\Unitary)$ on $I \oplus I$ in all nontrivial cases.

\begin{figure}[t]
\begin{align*}
  b & \defeq 0 \mid 1 \mid b+b \mid b \times b & \text{(value types)} \\
  t & \defeq b \isozh b & \text{(combinator types)} \\
  m & \defeq \nil \mid c_Z \cons m \mid c_{\phi} \cons m & \text{(amalgamations)}
\end{align*}
\begin{equation*}
  \frac{}{\nil \of b \isozh b} \qquad
  \frac{c \of b_1 \isoz b_2 \quad cs \of b_2 \isozh b_3}{c \cons cs \of b_1 
  \isozh b_3} \qquad
  \frac{c \of b_1 \isoh b_2 \quad cs \of b_2 \isozh b_3}{c \cons cs \of b_1 
  \isozh b_3}
\end{equation*}
\caption{\label{fig:pizh}Syntax and type rules of $\PizhLang$. The combinators $c_Z$ and $c_{\phi}$ are $\PiLang$ combinators (Fig.~\ref{fig:pi}) tagged with their sublanguage of origin.}
\end{figure}

\begin{figure}[t]
\begin{align*}
  m & \defeq \ldots \mid m \append m & \text{(derived amalgamations)} \\
  d & \defeq m \mid \cm{arr}_Z~c_Z \mid \cm{arr}_{\phi}~c_{\phi} \mid d \ggg d 
   & \text{(derived combinators)} \\
  & \mid \pid \mid \swapt \mid \assoct \mid \associt 
  \mid \unitet \mid \unitit \\
  & \mid \cm{first}~d \mid \cm{second}~d \mid d \arrprod d \mid \inv~d
\end{align*}
\caption{\label{fig:pizhp}Derived $\PizhLang$ constructs.}
\end{figure}

\subsection{$\PizhLang$: Syntax, Arrow Combinators, and Computational Universality}\label{sub:pizh-lang}
\label{sub:pizh}

As the two languages $\PizLang$ and $\PihLang$ share the same syntax,
their syntactic amalgamation in Fig.~\ref{fig:pizh} is rather
straightforward. We simply build sequences of expressions coming from
either language. To disambiguate amalgamations, we will annotate terms
from $\PizLang$ and $\PihLang$ by their language of origin and write,
\textit{e.g.}, $\xgate_{\phi}$ for the $\xgate$ isomorphism from
$\PihLang$. Further, we will consider the cons operator~$\cons$ to be
right associative, and use list notation such as $[\cm{swap}_{\phi}^+,
  \cm{swap}_Z^+]$ as syntactic sugar for the amalgamation
$\cm{swap}_{\phi}^+ \cons \cm{swap}_Z^+ \cons \nil$.

For convenience, we introduce the meta-operation $\cdot @ \cdot$ that
takes two amalgamations and forms the amalgamation given by their
concatenation, \textit{e.g.}, $[c_1,c_2] \append [c_3,c_4] =
[c_1,c_2,c_3,c_4]$. As such, any amalgamation can be uniquely
described as a finite heterogeneous list of terms from~$\PihLang$
and~$\PizLang$. But there is no need to reason about raw lists since
$\PizhLang$ is an \emph{arrow} over both $\PizLang$ and $\PihLang$
that lifts the underlying multiplicative structure to the combined
language (see Fig.~\ref{fig:pizhp} for the derived arrow constructs).
We recall that the construction only involves lifting \emph{products}
(via the arrow combinators $\cm{first}$, $\cm{second}$, and
$\arrprod$) and not also \emph{sums}; in other words, we merely define
an \emph{arrow} and not an \emph{arrow with
  choice}~\cite{hughes:arrows}. The reason is that doing so in any
meaningful way would require semantically identifying the sum
structures of $\PizLang$ and $\PihLang$, and that this, in turn, would
prevent quantum behaviours from emerging from the construction.

\begin{figure}
  \begin{equation*}
    \begin{array}{rclrclrcl}
      \multicolumn{9}{l}{\textbf{Semantics of new constructs}} \\
      &&& \sem{[c_1, \dots, c_n]} &=& [\sem{c_1}, \dots, \sem{c_n}]&&&\\[1.5ex]
      \multicolumn{9}{l}{\textbf{Derived identities}} \\
      \sem{\cm{arr}_Z~c} &=& \mathscr{E}_L(\sem{c}) &
      \sem{\cm{arr}_{\phi}~c} &=& \mathscr{E}_R(\sem{c}) & 
      \sem{d_1 \ggg d_2} &=& \sem{d_2} \circ \sem{d_1} \\
      \sem{\id} &=& [] &
      \sem{\swapt} &=& [\sigma_\otimes] \\
      \sem{\assoct} &=& [\alpha_\otimes] &
      \sem{\associt} &=& [\alpha_\otimes^{-1}] \\
      \sem{\unitet} &=& [\rho_\otimes] &
      \sem{\unitit} &=& [\rho_\otimes^{-1}] \\
      \sem{d_1 \arrprod d_2} &=& \sem{d_1} \otimes \sem{d_2} &
      \sem{\inv~d} &=& \sem{d}^\dagger \\
      \sem{\cm{first}~d} &=& \sem{d} \otimes \id &
      \sem{\cm{second}~d} &=& \id \otimes \sem{d} 
    \end{array}
  \end{equation*}
  \caption{The arrow semantics of $\PizhLang$.}
  \label{fig:pizh-arr-sem}
\end{figure}

The semantics of $\PizhLang$ is given in
Fig.~\ref{fig:pizh-arr-sem}. First the lifting of combinators builds
singleton lists:
\[
\cm{arr}_Z(c_Z) = [c_Z] \qquad\text{and}\qquad \cm{arr}_{\phi}(c_{\phi}) = [c_{\phi}]
\]
and composition of amalgamations is given by their concatenation:
\[
{cs}_1 \ggg {cs}_2 = cs_1 \append cs_2
\]
To define $\cm{first}$ we need to make use of meta-level recursion in order to traverse amalgamations. To do this, we notice that both $\PizLang$ and $\PihLang$ are trivially arrows, with arrow lifting given by the identity, arrow composition given by composition, and $\cm{first}~c$ given by $\cm{first}~c = c \times \pid$. As such, $\cm{first}$ can be defined in $\PizhLang$ by simply mapping this underlying combinator over the list, \textit{i.e.}, 
$$ \cm{first}~xs = \cm{map}~\cm{first}~xs$$ To derive $\cm{second}$,
we note that we can define all of the combinators relating to $\times$
(precisely: $\swapt, \assoct, \associt, \unitet, \unitit$) by lifting
them from either $\PizLang$ or $\PihLang$. It turns out not to matter
which we choose, as they are equivalent. Arbitrarily, we define their
liftings in $\PizhLang$ to be those from $\PizLang$, \textit{e.g.},
\[
\swapt = \cm{arr}_Z(\swapt_Z)
\]
and so on for the remaining ones. We can then derive $\cm{second}$ and $\arrprod$ in the usual way as:
\[
\cm{second} = \swapt \ggg \cm{first} \ggg \swapt \enspace \quad\text{and}\quad
xs \arrprod ys = \cm{first}~xs \ggg \cm{second}~ys
\]

We conclude this section by showing that there exists a particular
model in which $\PizhLang$ is computationally universal for quantum
circuits.

\begin{theorem}
  \label{thm:approx-universality-pizh}
  If $\phi$ is chosen to be $\pi/8$, the model of $\PizhLang$ is
  computationally universal for quantum circuits, \textit{i.e.},
  unitaries on Hilbert spaces of dimension $2^n$ (for any natural
  number $n$).
\end{theorem}
\begin{proof}
  Define a functor $\Aut_{R\phidpi}(\Unitary) \to \Unitary$ given
  simply by $A \mapsto A$ and $(R\phidpi_B)^{-1} \circ f \circ R\phidpi_A \mapsto
  (R\phidpi_B)^{-1} \circ f \circ R\phidpi_A$; this is strict monoidal and
  identity-on-objects. By Lem.~\ref{lem:monamalgfree}, this functor,
  along with the identity $\Unitary \to \Unitary$ uniquely define a
  (strict monoidal identity-on-objects) functor $\sem{\cdot} :
  \SymMonAmalg(\Unitary, \Aut_{R\phidpi}(\Unitary)) \to \Unitary$
  sending~$\PizLang$ programs to their usual interpretation lifted
  into unitaries, and $\PihLang$ programs to the meaning of the
  corresponding $\PizLang$ program conjugated by appropriate $R\phidpi_A$'s.
  Under this interpretation,
  $\sem{\cm{arr}_Z~\ccxgate}$ is the Toffoli
  gate, and $\sem{\cm{arr}_{\phi}~\xgate}$ is the Hadamard gate, while at
  the same time $\sem{c_1 \arrprod c_2} = \sem{c_1} \otimes \sem{c_2}$
  and $\sem{c_1 \ggg c_2} = \sem{c_2} \circ \sem{c_1}$, allowing
  parallel and sequential composition of gates. But then it follows by
  Thm.~\ref{thm:aharonov} that $\PizhLang$ is computationally
  universal.
\end{proof}

\section{$\SPizheLang$ from States and Effects}
\label{sec:spice}

The previous language $\PizhLang$ is not expressive enough to define
state preparations (and dually effect preparations). More generally,
the language imposes a rigid constraint on quantum circuits: the
number of wires must remain constant throughout the entire circuit: it
is not possible to temporarily allocate an ancilla, use it during a
subexpression, and discard it after. More importantly, the implicit
structures for cloning and joining classical states embedded in
$\PizhLang$ are not visible: making these explicit is necessary for us
to state constraints on their interaction from which quantum
properties will emerge. After a short exposition on cloning in quantum
computing, we develop the categorical construction to define classical
structures and ancillae and then use it to generalise $\PizhLang$ with
the notions of states $\ket{\cdot}$ and effects $\bra{\cdot}$ yielding
the language $\SPizheLang$.

\subsection{Classical Structures} 

There is some subtlety about copying in quantum computing. The
no-cloning theorem states that it is not possible to clone an
arbitrary quantum state. However, it is possible (and quite common) to
clone the subset of quantum states that are ``classical.'' For
example, the following maps are definable:
\[\begin{array}{rcl}
\ket{0} &\mapsto& \ket{00} \\
\ket{1} &\mapsto& \ket{11} \\
1/\sqrt{2} ~(\ket{0} + \ket{1}) &\mapsto& 1/\sqrt{2}~(\ket{00} + \ket{11})
\end{array}\]
If partial operations are allowed, then these maps are reversible:
that is to say, the classical clone maps are injective but not
surjective, so their inverses are only \emph{partial} injective
functions.

An important property of these classical clone maps is that their
behaviour is \emph{basis dependent}. In particular, the above maps
assume the clone operation is defined in the computational
$Z$-basis. If we instead use the $X$-basis $= \{ \ket{+}, \ket{-} \}$,
we get that cloning $\ket{+} = 1/\sqrt{2} ~(\ket{0} + \ket{1})$
produces $\ket{++} = 1/2 (\ket{00} + \ket{01} + \ket{10} + \ket{11})$
which is quite different from cloning the same state in the $Z$-basis.

As we will establish, the cloning operations in $\PizLang$ and
$\PihLang$ each satisfy the properties of classical structures
necessary for quantum behaviour to emerge in the next section.

\subsection{Monoidal Indeterminates}
\label{sub:ind}

At a categorical level, the only missing ingredient is to allow for
morphisms to manipulate ancilla systems $\mathscr{J}_N(U)$ generated
by a single object $N$. The first step in this construction is, given
a symmetric monoidal category $\cat{C}$, to define the monoidal
subcategory generated by taking arbitrary monoidal products of a
single generating object $N$ with itself, with only monoidal coherence
isomorphisms between them.
\begin{definition}
  Given a symmetric monoidal category $\cat{C}$ with a distinguished object $N$,
  define the symmetric monoidal category $\Gen_N(\cat{C})$ as follows.
  \begin{itemize}
    \item \textbf{Objects:} All finite monoidal products of objects generated 
    by $I$ and $N$, \textit{e.g.}, $I$, $N \otimes I$, $N \otimes (I \otimes N)$, etc.
    \item \textbf{Morphisms:} All finite monoidal products of morphisms 
    generated by identities as well as the monoidal coherence isomorphisms
    $\lambda_\otimes$, $\rho_\otimes$, $\sigma_\otimes$, $\alpha_\otimes$ and
    their inverses.
    \item \textbf{Composition and identities}: As in $\cat{C}$.
    \item \textbf{Monoidal structure}: As in $\cat{C}$.
  \end{itemize}
\end{definition}
This is clearly a symmetric monoidal subcategory of $\cat{C}$ with a strict
monoidal embedding $\mathscr{J}_N : \Gen_N(\cat{C}) \to \cat{C}$ given by the
identity on objects and morphisms. Putting this together:
\begin{definition}
  Define a symmetric monoidal category $\cat{C}_N$ as follows:
  \begin{itemize}
    \item \textbf{Objects:} As in $\cat{C}$.
    \item \textbf{Morphisms:} Morphisms $A \to B$ are equivalence
      classes of triples $[U,f,V]$ consisting of two objects $U$ and
      $V$ of $\Gen_N(\cat{C})$ and a morphism $A \otimes \mathscr{J}_N(U)
      \to B \otimes \mathscr{J}_N(V)$ under the equivalence $\sim$
      below.
    \item \textbf{Identities:} The identity $A \to A$ is the equivalence class 
    of $\id_{A \otimes I}$ (since $\mathscr{J}_N(I) = 
    I$).
    \item \textbf{Composition}: The composition of $[U,f,V]$ and $[W,g,X]$ with
    $f \colon A \otimes \mathscr{J}_N(U) \to B \otimes \mathscr{J}_N(V)$ and $g \colon B
    \otimes \mathscr{J}_N(W) \to C \otimes \mathscr{J}_N(X)$ is the equivalence
    class of the representative $A \otimes \mathscr{J}_N(U \otimes W) \to C \otimes
    \mathscr{J}_N(V \otimes X)$ in $\cat{C}$ given by 
    \[
    [U \otimes W,
    \alpha_\otimes \circ g \otimes \id_{\mathscr{J}_N(V)} \circ
    \alpha_\otimes^{-1} \circ \id_B \otimes \sigma_\otimes \circ \alpha_\otimes
    \circ f \otimes \id_{\mathscr{J}_N(W)} \circ \alpha_\otimes^{-1}, X \otimes
    V] \text.
    \]
    \item \textbf{Monoidal structure}: On objects as in $\cat{C}$. On
    morphisms, the monoidal product of $[U, f, V]$ and $[W, g, X]$ with $f \colon A
    \otimes \mathscr{J}_N(U) \to B \otimes \mathscr{J}_N(V)$ and $g \colon C \otimes \mathscr{J}_N(W) \to D \otimes \mathscr{J}_N(X)$ is $[U \otimes W,
    \vartheta^{-1} \circ f \otimes g \circ \vartheta, V \otimes X]$, where
    $\vartheta \colon (A \otimes C) \otimes (B \otimes D)
    \to (A \otimes B) \otimes (C \otimes D)$ is the evident isomorphism. Coherence isomorphisms $\beta$
    are given by $[I, \beta \otimes \id_I, I]$.
  \end{itemize}
  Define the equivalence relation $\sim$ as the least such satisfying (for all $U$ and $V$)
  \begin{equation}
    f \otimes \mathscr{J}_N(\id_U) \sim f \otimes \mathscr{J}_N(\id_V)
  \end{equation}
  and $f \sim g$ if there exist mediators $m \colon U \to U'$ and $n \colon V \to V'$ in 
  $\Gen_N(\cat{C})$ making the square below commute.
  \begin{equation}
    \begin{tikzcd}[ampersand replacement=\&]
    	{A \otimes \mathscr{J}_N(U)} \& {B \otimes \mathscr{J}_N(V)} \\
    	{A \otimes \mathscr{J}_N(U')} \& {B \otimes \mathscr{J}_N(V')}
    	\arrow["f", from=1-1, to=1-2]
    	\arrow["g"', from=2-1, to=2-2]
    	\arrow["{\id \otimes \mathscr{J}_N(m)}"', from=1-1, to=2-1]
    	\arrow["{\id \otimes \mathscr{J}_N(n)}", from=1-2, to=2-2]
    \end{tikzcd}
  \end{equation}
\end{definition}
This describes a symmetric monoidal
category~\cite{hermidatennent:indeterminates}. Note that the two
clauses in the equivalence relation are necessary precisely to ensure
uniqueness of identities and associativity of composition. For a given
morphism $[X, f, Y]$, we will collectively denote the objects $X$ and
$Y$ as the \emph{ancilla system} of $[X, f, Y]$.

Importantly, this also defines an arrow over $\cat{C}$ in the form of a strict
monoidal functor:

\begin{proposition}
  There is a strict monoidal functor $\mathscr{F}_N \colon \cat{C} \to \cat{C}_N$
  (for any choice of $N$) given by $A \mapsto A$ on objects and $f \mapsto [I,
  f \otimes \id_I, I]$ on morphisms.
\end{proposition}
\begin{proof}
  \citet{hermidatennent:indeterminates} show that this functor is
  strong monoidal; it follows readily that it is strict monoidal when
  the functor $\Gen_N(\cat{C}) \to \cat{C}$ is.
\end{proof}

The category $\cat{C}_N$ contains all morphisms $f$ of $\cat{C}$ by lifting them into the trivial ancilla system $[I, f \otimes \id_I, I]$. However, it also adds a state $I \to N$ (as the equivalence class of $\sigma_\otimes : I \otimes N \to N \otimes I$, \textit{i.e.}, $[N, \sigma_\otimes, I]$) and an effect $N \to I$ (as the equivalence class of $\sigma_\otimes : N \otimes I \to I \otimes N$, \textit{i.e.}, $[I, \sigma_\otimes, N]$). 

One may reasonably wonder whether this construction adds more than this unique state and effect; after all, there are many more ancilla systems than just $I$ and $N$. The answer, informally, is no. The trivial ancilla system (\textit{e.g.}, $[I, f, I]$) is needed to account for morphisms that do not use states or effects at all, but multiple occurrences of $I$ in an ancilla system (\textit{e.g.}, $I \otimes I$ or $(I \oplus I) \otimes I$) are vacuous, as we are able to remove them by mediating with the left or right unitor using the equivalence relation. This leaves only ancilla systems of the form $(I \oplus I)^{\otimes n}$, where precise bracketing does not matter as we are able to mediate by arbitrary combinations of associators. These ancilla systems represent multiple uses of the unique state or effect.

The more formal answer to this question is that the functor $\mathscr{F}_N$ is universal with this property:
\begin{proposition}\label{prop-ht-univ}
  Given any symmetric monoidal category $\cat{D}$ outfitted with distinguished
  morphisms $I \to N$ and $N \to I$, and strict monoidal functor $F : \cat{C}
  \to \cat{D}$, there is a unique strict monoidal functor $\widehat{F} : \cat{C}_N
  \to \cat{D}$ making the triangle below commute.
  \[\begin{tikzcd}[ampersand replacement=\&]
  	{\cat{C}} \&\& {\cat{C}_N} \\
  	\&\& {\cat{D}}
  	\arrow["F"', from=1-1, to=2-3]
  	\arrow["{\mathscr{F}_N}", from=1-1, to=1-3]
  	\arrow["{\widehat{F}}", dashed, from=1-3, to=2-3]
  \end{tikzcd}\]
\end{proposition}
\begin{proof}
  The construction is a dual pair of \emph{monoidal indeterminates}
  constructions of \citet{hermidatennent:indeterminates} (using the simplified
  form for the equivalence due to
  \citet[Def.~8]{andresmartinezheunenkaarsgaard:universal}, so the 
  theorem follows from Thm~2.9 of \citet{hermidatennent:indeterminates},
  noting that this straightforwardly extends from strong monoidal functors to
  strict monoidal ones when the inclusion $\cat{C}_N \to \cat{C}$ is strict
  monoidal.
\end{proof}

Note that the morphisms $I \to N$ and $N \to I$ are considered part of the
\emph{structure} of $\cat{D}$ in the above. As such, in as far as $\cat{D}$ has
other choices of morphisms $I \to N$ and $N \to I$, the theorem states that for
any such choice of morphisms and strict monoidal functor $\cat{C} \to \cat{D}$,
there is a unique strict monoidal functor $\cat{C}_N \to \cat{D}$.

Observe further that $\cat{C}_N$ is a dagger category when $\cat{C}$ is, with $[X, f, Y]^\dagger = [Y, f^\dagger, X]$. With this dagger structure, the adjoint of the state $I \to N$ is the effect $N \to I$ and vice versa.

\subsection{Models of $\SPizheLang$: $\SymMonAmalg(\zlang{\Unitary}, \hlang{\Aut_{R\phid}(\Unitary))}_{I\oplus I}$ and $\cat{Contraction}$}
\label{sub:spicemodel}

As explained in the previous section, given a model $\cat{C}$ of
$\PizhLang$, we can define the category $\cat{C}_{N}$ which extends
$\cat{C}$ to model a language equipped with a state $I \to N$ and an
effect $N \to I$. In order to model quantum computation, we specialize
this free model by fixing the semantics of the state and effect to
correspond to the traditional basis vector $\ket{0}$ and dual vector
$\bra{0}$. This is achieved by choosing $N=I \oplus I$ making
$\SymMonAmalg(\zlang{\Unitary},
\hlang{\Aut_{R\phid}(\Unitary))}_{I\oplus I}$ a model of
$\SPizheLang$. This model embeds in $\Contraction$, the universal
dagger rig category containing all unitaries, states, and
effects~\cite{andresmartinezheunenkaarsgaard:universal} as follows.

Recall how we gave $\PizhLang$ unitary semantics via a
(identity-on-objects) strict monoidal functor $\sem{-} :
\SymMonAmalg(\Unitary, \Aut_{R\phid}(\Unitary))$. Recall further that
the category $\Contraction$ of finite-dimensional Hilbert spaces and
contractions contains all states as morphisms $I \to A$ (as these are
isometries), all effects as morphisms $A \to I$ (as these are
coisometries). Since all unitaries are contractions we get an
inclusion functor $\Unitary \to \Contraction$ (which is easily seen
to be a strict monoidal dagger functor), and precomposing with the
strict monoidal functor $\sem{-} \colon \SymMonAmalg(\Unitary,
\Aut_{R\phid}(\Unitary)) \to \Unitary$ yields a functor:
\[
  \SymMonAmalg(\Unitary, \Aut_{R\phid}(\Unitary)) \to \Contraction 
\]
Now, choosing $\ket{0}$ as the distinguished state $I \to I \oplus I$, and
$\bra{0}$ as the distinguished effect $I \oplus I \to I$, by
Prop.~\ref{prop-ht-univ} we get a unique extended functor
\[
  \sem{-} \colon \SymMonAmalg(\Unitary, \Aut_{R\phid}(\Unitary))_{I \oplus I} \to \Contraction
\]
which assigns semantics to $\SPizheLang$ in $\Contraction$.

\begin{figure}[t]
\begin{align*}
  b & \defeq 0 \mid 1 \mid b+b \mid b \times b & \text{(value types)} \\
  n & \defeq 1 \mid 1+1 \mid n \times n & \text{(ancilla types)} \\
  t & \defeq b \arzh b & \text{(combinator types)} \\
  p & \defeq \cm{lift}~m & \text{(primitives)} 
\end{align*}
\begin{equation*}
  \frac{xs \of b_1 \times n_1 \isozh b_2 \times n_2}{
  \cm{lift}~xs \of b_1 \arzh b_2}
\end{equation*}
\caption{\label{fig:spizhe}$\SPizheLang$ syntax and type rules. Amalgamations $m$ are defined in Figs.~\ref{fig:pizh} and~\ref{fig:pizhp}.}
\end{figure}

\begin{figure}[t]
\begin{align*}
  d & \defeq p \mid \cm{arr}~m \mid d \ggg d \mid \cm{first}~d \mid 
  \cm{second}~d \mid d \arrprod d  & \text{(derived combinators)} \\
  & \mid \pid \mid \swapt \mid \assoct \mid \associt 
  \mid \unitet \mid \unitit \\
  & \mid \inv~d \mid \cm{zero} \mid \cm{assertZero} 
\end{align*}
\caption{\label{fig:spizhep}Derived $\SPizheLang$ constructs.}
\end{figure}

\begin{figure}
  \begin{equation*}
    \begin{array}{rclrclrcl}
      \multicolumn{9}{l}{\textbf{Semantics of new constructs}} \\
      &&& \sem{\cm{lift}~xs} &=& [\sem{n_1}, \sem{xs}, \sem{n_2}] &&&\\[1.5ex]
      \multicolumn{9}{l}{\textbf{Derived identities}} \\
      \sem{\cm{arr}~m} &=& \mathscr{F}_{I \oplus I}(\sem{m}) &
      \sem{d_1 \ggg d_2} &=& \sem{d_2} \circ \sem{d_1} \\
      \sem{\id} &=& [I, \id, I] &
      \sem{\swapt} &=& [I, \sigma_\otimes \otimes \id, I] \\
      \sem{\assoct} &=& [I, \alpha_\otimes \otimes \id, I] &
      \sem{\associt} &=& [I, \alpha_\otimes^{-1} \otimes \id, I] \\
      \sem{\unitet} &=& [I, \rho_\otimes \otimes \id, I] &
      \sem{\unitit} &=& [I, \rho_\otimes^{-1} \otimes \id, I] \\
      \sem{\cm{first}~d} &=& \sem{d} \otimes \id &
      \sem{\cm{second}~d} &=& \id \otimes \sem{d} &
      \sem{d_1 \arrprod d_2} &=& \sem{d_1} \otimes \sem{d_2} \\
      \sem{\inv~d} &=& \sem{d}^\dagger &
      \sem{\cm{zero}} &=& [I \oplus I, \sigma_\otimes, I] &
      \sem{\cm{assertZero}} &=& [I, \sigma_\otimes, I \oplus I]
    \end{array}
  \end{equation*}
  \caption{The arrow semantics of $\SPizheLang$.}
  \label{fig:spizhe-arr-sem}
\end{figure}

\subsection{Syntax, States, and Effects}
\label{sub:spice}

In $\SPizheLang$, we allow the creation and discarding of a restricted
set of values of ancilla types. As given in Fig.~\ref{fig:spizhe}, the
ancilla types are restricted to be collections of bits. The construct
$\cm{lift}$ allows the discarding of some ancilla $n_1$ and the
creation of some ancilla $n_2$. Like in the previous section, the new
language defines an arrow over $\PizhLang$ with the derived
combinators in Fig.~\ref{fig:spizhep}. Most notably, the language
includes two new derived constructs $\cm{zero}$ and $\cm{assertZero}$
whose semantics are $\ket{0}$ and $\bra{0}$ respectively.

We summarise the arrow semantics of $\SPizheLang$ in
Figure~\ref{fig:spizhe-arr-sem}. To see that this is an arrow, we must
define $\cm{arr}$, $\ggg$, and $\cm{first}$. Bringing combinators from
$\PizhLang$ into $\SPizheLang$ is straightforwardly done by adding the
trivial ancilla $1$ to both the input and output,
$$
\cm{arr}~m = \cm{lift}(\unitet \ggg m \ggg \unitit) \enspace.
$$
This allows us to lift the isomorphisms $\pid$, $\swapt$, $\assoct$, $\associt$, $\unitet$, and $\unitit$ of $\PizhLang$ simply by applying $\cm{arr}$ to them.
To compose lifted $\PizhLang$ terms $m : b_1 \times n_1 \isozh b_2 \times n_2$ and $p : b_2 \times n_3 \isozh b_3 \times n_4$, since ancillae are closed under products, we can form this as the lifting of a term of type $b_1 \times (n_1 \times n_3) \isozh b_3 \times (n_4 \times n_2)$, namely
\begin{align*}
(\cm{lift}~m) \ggg (\cm{lift}~p) & = \cm{lift}(\associt \ggg \cm{first}~m \ggg \assoct \ggg \cm{second}~\swapt \ggg \\
& \phantom{= \cm{lift}(} \associt \ggg \cm{first}~p \ggg \assoct) \enspace.
\end{align*}
Then, $\cm{first}$ can be defined using $\cm{first}$ in $\PizhLang$, since this allows us to extend a lifted term of type $b_1 \times n_1 \isozh b_2 \times n_2$ to one of type $(b_1 \times n_1) \times b_3 \isozh (b_2 \times n_2) \times b_3$, so we need only swap the ancillae back into the rightmost position from there, \textit{i.e.},
\begin{align*}
\cm{first}(\cm{lift}~m) & = \cm{lift}(\assoct \ggg \cm{second}~\swapt \ggg \associt \ggg \cm{first}~m \ggg \\
& \phantom{= \cm{lift}(}\assoct \ggg \cm{second}~\swapt \ggg \assoct) \enspace.
\end{align*}
In turn, $\cm{second}$ and $\arrprod$ are derived exactly as in $\PizhLang$. Inversion is simple since lifted terms are symmetric in having an ancilla type on both their input and output, so we have
$$
\inv(\cm{lift}~m) = \cm{lift}(\inv~m) \enspace.
$$ Finally, the state $\cm{zero}$ and effect $\cm{assertZero}$ exist
as the lifting of $\swapt \of 1 \times (1+1) \isozh (1+1) \times 1$
and $\swapt \of (1+1) \times 1 \isozh 1 \times (1+1)$, \textit{i.e.},
\[
\cm{zero} = \cm{lift}(\swapt) \of 1 \arzh 1+1 \qquad\text{and}\qquad
\cm{assertZero} = \cm{lift}(\swapt) \of 1+1 \arzh 1 \enspace,
\]
bringing the state into and out of focus respectively. A pleasant
consequence of these definitions is that $\inv(\cm{zero}) =
\cm{assertZero}$ and vice versa. More generally, states and effects in
$\SPizheLang$ satisfy the following properties.

\begin{proposition}[Classical Structures for $\PizLang$ and $\PihLang$ and their execution laws]\label{prop:classexeclaws}
To avoid clutter, we will implicitly lift $\PizLang$ and $\PihLang$
gates to $\SPizheLang$, writing $c_Z$ for $\cm{arr}~(\cm{arr}_Z~c)$
and $c_\phi$ for $\cm{arr}~(\cm{arr}_\phi~c)$. Introduce the following
abbreviations:
\begin{align*}
\cm{copy}_Z & =
  \unitit \ggg \pid \arrprod \cm{zero} \ggg \cxgate_Z & 
\cm{copy}_X & =
  \xgate_{\phi} \ggg \cm{copy}_Z \ggg \xgate_{\phi} \arrprod \xgate_{\phi} \\ 
\cm{one} &= \cm{zero} \ggg \xgate_Z &
\cm{assertOne} &= \xgate_Z \ggg \cm{assertZero} 
\end{align*}
The following equations are satisfied in $\SPizheLang$:
\begin{align*}
  \cm{copy}_Z \ggg (\pid \arrprod \cm{copy}_Z) 
  &~=~  
  \cm{copy}_Z \ggg (\cm{copy}_Z \arrprod \pid) \ggg \cm{assoc}^\times
  \\
  \cm{copy}_Z \ggg \swapt
  &~=~ 
  \cm{copy}_Z
  \\
  \cm{copy}_Z \ggg (\inv\; \cm{copy}_Z)
  &~=~
  \id
  \\
  (\cm{copy}_Z \arrprod \pid) \ggg (\pid \arrprod \inv\; \cm{copy}_Z)
  &~=~ 
  (\pid \arrprod \cm{copy}_Z) \ggg (\inv\; \cm{copy}_Z \arrprod \pid) \\
\\
  \cm{copy}_X \ggg (\pid \arrprod \cm{copy}_X) 
  &~=~ 
  \cm{copy}_X \ggg (\cm{copy}_X \arrprod \pid) \ggg \cm{assoc}^\times
  \\
  \cm{copy}_X \ggg \swapt
  &~=~ 
  \cm{copy}_X
  \\
  \cm{copy}_X \ggg (\inv\; \cm{copy}_X)
  &~=~
  \id
  \\
  (\cm{copy}_X \arrprod \pid) \ggg (\pid \arrprod \inv\; \cm{copy}_X)
  &~=~ 
  (\pid \arrprod \cm{copy}_X) \ggg (\inv\; \cm{copy}_X \arrprod \pid) \\
\\
  \cm{zero} \ggg \cm{assertZero} &~=~ \pid \\
  \cm{zero} \arrprod \pid \ggg \ctrlgate~c &~=~ \cm{zero} \arrprod \pid \\
  \cm{one} \arrprod \pid \ggg \ctrlgate~c &~=~ \cm{one} \arrprod c \\
  \cm{zero} \ggg \xgate_{\phi} \ggg \cm{assertOne} &~=~ \cm{one} \ggg \xgate_{\phi} \ggg \cm{assertZero}
\end{align*}
\end{proposition}
\begin{proof}
The first two groups state that $\cm{copy}_Z$ and $\cm{copy}_X$ are
each a classical cloning map for the relevant basis: the $Z$-basis in
the case of $\cm{copy}_Z$ and some rotated basis depending on $\phi$
for $\cm{copy}_X$. Because we fixed the semantics of $\PizLang$ to be
the standard semantics in $\Unitary$ without any rotation, the action
of $\cm{copy}_Z$ on input $v$ is to produce the pair $(v,v)$.  The
rotated version $\cm{copy}_X$ has the same semantics but in another
basis. The last group of equations are the \emph{execution equations},
which are so named as they describe how the states (and, by duality,
effects) interact with program execution. The first simply shows that
preparing the zero state and then asserting it does nothing at
all. The remaining equations define how states (and, by dualising the
equations, effects) must interact with control, and with one another:
e.g., the second equation shows that passing $\ket{0}$ to a control
line prevents the controlled program from being executed, while the
third shows that passing $\ket{1}$ on a control line executes the
controlled program.
\end{proof}

\subsection{Computational Universality}

Like in the previous section, we can conclude that there exists a
particular model in which $\SPizheLang$ is computationally universal for
quantum circuits equipped with arbitrary states and effects.

\begin{theorem}[Expressivity]\label{thm:approx-universality-spizhe}
  If $\phi$ is chosen to be $\pi/8$, the model of $\SPizheLang$ is
  computationally universal for quantum circuits equipped with
  arbitrary states and effects.
\end{theorem}
The more technical presentation of this theorem is the following. Say
that a \emph{preparation of states} on a $2^n$-dimensional Hilbert
space is a tensor product $s_1 \otimes \cdots \otimes s_n$, where each
$s_i$ is either a state or an identity. Dually, a \emph{preparation of
  effects} is the adjoint $s_1^\dagger \otimes \cdots \otimes
s_n^\dagger$ to a preparation of states $s_1 \otimes \cdots \otimes
s_n$. The theorem then states that $\SPizheLang$ is approximately
universal for contractions $\hilbspace \to \hilbspacep$ between
Hilbert spaces \hilbspace\ (of dimension $2^n$) and \hilbspacep\ (of
dimension $2^m$) of the form $S \mathcal{U} E$, where $S$ is a
preparation of states, $\mathcal{U}$ is unitary, and $E$ is a
preparation of effects.
\begin{proof}
  Let $s_1 \otimes \cdots \otimes s_n$ be some state preparation,
  $\mathcal{U}$ be some unitary, and $t_1^\dagger \otimes \cdots
  \otimes t_n^\dagger$ be some effect preparation.  In the state
  preparation $s_1 \otimes \cdots \otimes s_n$, for each non-identity
  $s_i$, choose some unitary $S_i$ mapping $\ket{0}$ to
  $s_i$. Likewise, in the effect preparation $t_1^\dagger \otimes
  \cdots \otimes t^\dagger_n$, choose for each non-identity
  $t_i^\dagger$ a unitary $T_i^\dagger$ mapping $\bra{0}$ to
  $t_i^\dagger$. Produce now a state preparation $s_i' = s_1' \otimes
  \cdots \otimes s_n'$ where $s_i' = \ket{0}$ if $s_i$ is a state, and
  $s_i' = \id$ if $s_i$ is an identity. Produce an effect preparation
  $t_1' \otimes \cdots \otimes t_n'$ similarly. Notice that $(S_1
  \otimes \cdots \otimes S_n)(s_1' \otimes \cdots \otimes s_n') = s_1
  \otimes \cdots \otimes s_n$ and $(T_1 \otimes \cdots \otimes
  T_n)(t_1' \otimes \cdots \otimes t_n') = t_1 \otimes \cdots \otimes
  t_n$.  However, the state preparation $s_1' \otimes \cdots \otimes
  s_n'$ involves only identities and $\ket{0}$, and so has a direct
  representation in $\PizhLang$ as a product of a number of $\pid$ and
  $\cm{zero}$ terms---call the resulting term~$p$. Likewise, $t_1'
  \otimes \cdots \otimes t_n'$ has a direct representation in
  $\PizhLang$ as a product of a number of $\pid$ and $\cm{assertZero}$
  terms---call the resulting term $q$. Finally, by
  Thm.~\ref{thm:approx-universality-pizh} we can approximate the
  unitary $ (S_1 \otimes \cdots \otimes S_n) \mathcal{U} (T_1^\dagger
  \otimes \cdots \otimes T_n^\dagger)$ by some $\PizhLang$ term
  $u$. But then $p \ggg u \ggg q$ approximates $(s_1 \otimes \cdots
  \otimes s_n) \mathcal{U} (t_1^\dagger \otimes \cdots \otimes
  t_n^\dagger)$.
\end{proof}

\noindent Though the proof above may at first glance appear to be non-constructive (as it involves choice among unitaries), we note that a unitary in finite dimension mapping $\ket{0}$ to some $\ket{v}$ (of norm $1$) can be constructed using the Gram-Schmidt process.

\section{Canonicity and Quantum Computational Universality}
\label{sec:canonicity}

The section proves the main result of the paper.

So far, we have built a particular model of $\SPizheLang$ in $\Contraction$
and proved that by imposing $\phi=\pi/8$, we get a computationally
universal quantum programming language. In fact, it turns out that we
have a much stronger result. \emph{Any model of $\SPizheLang$ in
  $\Contraction$ satisfying the equations for classical structures and
  their execution laws defined in Prop.~\ref{prop:classexeclaws} as
  well as the complementarity equation in Def.~\ref{def:comp} is
  computationally universal!}

\begin{definition}[Complementarity]\label{def:comp}
To avoid clutter, we will implicitly lift $\PizLang$ and $\PihLang$
gates to $\SPizheLang$, writing $c_Z$ for $\cm{arr}~(\cm{arr}_Z~c)$
and $c_\phi$ for $\cm{arr}~(\cm{arr}_\phi~c)$. The complementarity law
requires the following identity:
\[\begin{array}{cl}
&  (\cm{copy}_Z \arrprod \pid)
  \ggg
  \assoct
  \ggg \\
&  (\pid \arrprod (\inv\;\cm{copy}_X))
  \ggg
  (\pid \arrprod \cm{copy}_X)
  \ggg
  \associt
  \ggg \\
&  ((\inv\;\cm{copy}_Z) \arrprod \pid) \\
  ~=~
&  \pid
\end{array}\]
\end{definition}
Recall that the complementarity law is equivalent to the two bases
being mutually unbiased in the usual sense.

To show the canonicity theorem, we rely on the following
characterisation of orthonormal bases complementary to the $Z$-basis
where the unitary change of basis is involutive:
\begin{proposition}\label{prop:charhad}
  Every orthonormal basis $\{\ket{b_1}, \ket{b_2}\}$ on $\mathbb{C}^2$
    which is complementary to the $Z$-basis, and for which the
    associated change of basis unitary is involutive, is on one of the
    following two forms for $\theta \in [0, 2\pi)$:
  \begin{align*}
    \ket{b_1} = \tfrac{1}{\sqrt{2}}\left(\begin{smallmatrix}
      1 \\ e^{-i \theta}
    \end{smallmatrix}\right) \qquad
    \ket{b_2} = \tfrac{1}{\sqrt{2}}\left(\begin{smallmatrix}
      e^{i \theta} \\ -1
    \end{smallmatrix}\right)
  \end{align*}
  or
  \begin{align*}
    \ket{b_1} = \tfrac{1}{\sqrt{2}}\left(\begin{smallmatrix}
      -1 \\ e^{-i \theta}
    \end{smallmatrix}\right) \qquad
    \ket{b_2} = \tfrac{1}{\sqrt{2}}\left(\begin{smallmatrix}
      e^{i \theta} \\ 1
    \end{smallmatrix}\right) \enspace.
  \end{align*}
\end{proposition}
\begin{proof}
  Let $\{\ket{b_1}, \ket{b_2}\}$ be an orthonormal basis on
  $\mathbb{C}^2$ complementary to the $Z$-basis. Since we then have
  \begin{align*}
    \braket{0|b_1}\braket{b_1|0} = \braket{1|b_1}\braket{b_1|1} & =
    \tfrac{1}{2} \\
    \braket{0|b_2}\braket{b_2|0} = \braket{1|b_2}\braket{b_2|1} & =
    \tfrac{1}{2}
  \end{align*}
  by complementarity, the only freedom in $\{\ket{b_1}, \ket{b_2}\}$
  is in the choice of phases $p_1, \dots, p_4$ in
  \begin{align*}
    \ket{b_1} = \tfrac{1}{\sqrt{2}}\left(\begin{smallmatrix}
      p_1 \\ p_2
    \end{smallmatrix}\right) \qquad
    \ket{b_2} = \tfrac{1}{\sqrt{2}}\left(\begin{smallmatrix}
      p_3 \\ p_4
    \end{smallmatrix}\right) \enspace.
  \end{align*}
  Its associated change of basis is given by the matrix
  $$
  U = \frac{1}{\sqrt{2}}\begin{pmatrix}
    p_1 & p_3 \\ p_2 & p_4
  \end{pmatrix}
  $$ which is seen to be unitary and further assumed to be involutive,
  so Hermitian. By $U$ Hermitian we have $p_1 =
  \overline{p_1}$, $p_2 = \overline{p_3}$, and $p_4 = \overline{p_4}$,
  so in particular $p_1 = \pm 1$ and $p_4 = \pm 1$. By orthogonality
  \begin{align*}
    \braket{b_1|b_2} = 0
    \quad & \Leftrightarrow \quad p_1 \overline{p_3} + p_2 \overline{p_4} = 0 \\
    \quad & \Leftrightarrow \quad p_1 p_2 = -p_4 p_2 \\
    \quad & \Leftrightarrow \quad p_1 p_2 \overline{p_2}
    = -p_4 p_2 \overline{p_2} \\
    \quad & \Leftrightarrow \quad p_1 = -p_4
  \end{align*}
  so in fact either $p_1 = 1$ and $p_4 = -1$ or the other way around.
  Writing $p_2 = e^{i \theta}$ and $p_3 = \overline{p_2} = e^{-i \theta}$ then
  gives us the following two possible forms for $U$,
  $$
  U = \frac{1}{\sqrt{2}} \begin{pmatrix}
    1 & e^{i \theta} \\ e^{-i \theta} & -1
  \end{pmatrix} \quad\text{or}\quad
  U = \frac{1}{\sqrt{2}} \begin{pmatrix}
    -1 & e^{i \theta} \\ e^{-i \theta} & 1
  \end{pmatrix}
  $$
  for $\theta \in [0,2\pi)$, proving the claim since the
    columns of $U$ were given by $\ket{b_1}$ and $\ket{b_2}$.
\end{proof}
\begin{theorem}[Canonicity]\label{thm:canonicity}
  If a categorical semantics $\sem{-}$ for $\SPizheLang$ in
  $\Contraction$ satisfies the classical structure laws and the
  execution laws (defined in Prop.~\ref{prop:classexeclaws}) and the
  complementarity law (Def.~\ref{def:comp}), then it is
  computationally universal. Specifically, it must be the semantics of
  Sec.~\ref{sub:spicemodel} with the semantics of $\xgate_\phi$ being
  the Hadamard gate (up to conjugation by $X$ and $Z$) and:
  \begin{align*}
    \sem{\cm{copy}_Z} & \colon \ket{i} \mapsto \ket{ii} &
    \sem{\cm{zero}} & = \ket{0} \\
    \sem{\cm{copy}_X} & \colon \ket{\pm} \mapsto \ket{\pm\pm} &
    \sem{\cm{assertZero}} & = \bra{0}
  \end{align*}
  up to a global unitary.
\end{theorem}
\begin{proof}
  Observe that $\sem{I+I}=\mathbb{C}^2$ must be the
  qubit.  Without loss of generality, we may assume that $\PizLang$
  has the usual semantics in the computational ($Z$) basis---this is
  the freedom that the global unitary affords.
  
  The execution equations ensure that
  $\{\sem{\cm{zero}},\sem{\cm{one}}\}$ and
  $\{\sem{\cm{plus}},\sem{\cm{minus}}\}$ are copyable by
  $\sem{\cm{copy}_Z}$ and $\sem{\cm{copy}_X}$ respectively, while the
  complementarity equations further ensure that
  $\{\sem{\cm{zero}},\sem{\cm{one}}\}$ and
  $\{\sem{\cm{plus}},\sem{\cm{minus}}\}$ form complementarity
  orthonormal bases for $\mathbb{C}^2$. By assumption
  $\sem{\cm{zero}}$ and $\sem{\cm{one}}$ form the $Z$-basis, so the
  only possibility for $\{\sem{\cm{plus}},\sem{\cm{minus}}\}$ is as
  an orthonormal basis complementary to the $Z$-basis. 

  Since $\sem{\cm{arr}_\phi~\swapp}$ is the symmetry of a symmetric
  monoidal category it is involutive, and by the complementarity
  equations it is the change of basis unitary between orthonormal
  bases. It follows then by Proposition~\ref{prop:charhad} that
  $$
  \sem{\cm{arr}_\phi~\swapp} = \frac{1}{\sqrt{2}} \begin{pmatrix}
    1 & e^{i \theta} \\ e^{-i \theta} & -1
  \end{pmatrix} \qquad\text{or}\qquad
  \sem{\cm{arr}_\phi~\swapp} = \frac{1}{\sqrt{2}} \begin{pmatrix}
    -1 & e^{i \theta} \\ e^{-i \theta} & 1
  \end{pmatrix}
  $$
  for some $\theta \in [0,2\pi)$. However, the execution equation
  $$
    \cm{zero} \ggg \xgate_\phi \ggg \cm{assertOne} =
    \cm{one} \ggg \xgate_\phi \ggg \cm{assertZero}
  $$
  translates to the requirement that
  $$
  e^{i \theta} = \bra{0}\sem{\cm{arr}_\phi~\swapp}\ket{1} =
  \bra{1}\sem{\cm{arr}_\phi~\swapp}\ket{0} = e^{-i \theta}
  $$
  in turn implying $e^{i \theta} = \pm 1$. This leaves
  $$
  \frac{1}{\sqrt{2}}
  \begin{pmatrix}
    1 & 1 \\ 1 & -1
  \end{pmatrix}, \qquad
  \frac{1}{\sqrt{2}}
  \begin{pmatrix}
    1 & -1 \\ -1 & -1
  \end{pmatrix}, \qquad
  \frac{1}{\sqrt{2}}
  \begin{pmatrix}
    -1 & 1 \\ 1 & 1
  \end{pmatrix}, \quad\text{and} \quad
  \frac{1}{\sqrt{2}}
  \begin{pmatrix}
    -1 & -1 \\ -1 & 1
  \end{pmatrix} \qquad
  $$
  as the only possibilities for $\sem{\cm{arr}_\phi~\swapp}$, which
  are precisely Hadamard up to conjugation by $Z$ and/or $X$. Either
  way, this is a real basis-changing single-qubit unitary, so
  computationally universal in conjunction with the Toffoli gate
  (which is expressible in $\PizLang$) by Theorem~\ref{thm:aharonov}.
\end{proof}

Observe that the Hadamard gate could already be expressed in
$\PizhLang$ with the appropriate choice of $\phi$, without states and
effects. The latter were only needed to impose equations on
$\PizhLang$ which essentially forces $\phi$ to be $\pi/8$.

\section{\qpi: Examples and Reasoning}
\label{sec:reasoning}
\label{sec:measurement}

Having shown that quantum behaviour emerges from two copies of a
classical reversible programming language mediated by the complementarity
equation, we now illustrate that developing quantum programs similarly only
needs classical principles augmented with the complementarity equation, and
some forms of reasoning can similarly be reduced. Before we
proceed however, we present a sanitised version of \SPizheLang\ that
fixes $\phi=\pi/8$, that hides some of the constructs that were only
needed for the intermediate steps, and that uses Agda syntax for ease
of experimentation and for providing machine-checked proofs of
equivalences. The code shown here is not self-contained, and the
previously mentioned repository should be consulted for the details.
Nevertheless this paper is ``literate Agda'' in the sense that we
extract the Agda code (using the \LaTeX package \texttt{minted})
directly from the sources.

\subsection{\qpi: Syntax and Terms}

The public interface of \qpi\ consists of two layers: the core
reversible classical language \PiLang\ (of Fig.~\ref{fig:pi}) and the
arrow layer (presented in Sec.~\ref{sec:story}). We reproduce these
below using the notation in our Agda specification.

The \qpi\ types are directly collected in an Agda datatype:

\agdafrag{Pi/Types.agda}{15}{19}

\noindent Since commutative monoids are used multiple times, their
definition is abstracted in a structure \agda{CMon} that is
instantiated twice:

\agdafrag{Pi/Language.agda}{21}{26}  

\noindent The \PiLang-combinators are encoded in a type family:

\agdafrag{Pi/Language.agda}{28}{38}  

\noindent Finally, the syntax and types of the \qpi\ combinators are encoded in
another type family that uses another instance of our commutative
monoid \agda{CMon}:

\agdafrag{QPi/Syntax.agda}{22}{34}

\noindent In the following, we will refer to common gates and states
which we collect here. The definitions are a straightforward transcription
into Agda of the ones in previous sections. Below $\Pi$ refers to the
module that (abstractly) defines \PiLang-combibators.

\agdafrag{QPi/Terms.agda}{38}{56}

\noindent And so, as expected, the \agda{X} gate and the \agda{H} gate are both
lifted versions of $\swapp$ from the underlying definition of $\Pi$.
The classical gates of \PiLang\ and their controlled versions are
lifted using \agda{arrZ}.

\subsection{Proving Simple Equivalences}

The laws of classical structures, the execution equations, and the
complementarity law, combined with the conventional laws for arrows and 
monoidal and rig categories, allow us to reason about \qpi\ programs
at an abstract extensional level that eschews complex numbers,
vectors, and matrices.

As a first demonstration, we can prove that the \agda{X} and \agda{H}
gates are both involutive:

\agdafrag{Reasoning.agda}{63}{71}

\noindent The first proof uses Agda's equational style where each step
is justified by one of the \qpi\ equivalences (see full code). The
proof starts by moving ``under the arrow'' exposing the underlying
$\swapp$ gate, using the fact that $\swapp$ is an involution, and then
lifting the equivalence back through the arrow. The proof for
\agda{H}, presented as a sequence of equivalences, applies the same
strategy to the other arrow.  In later examples, we also use some
additional reasoning combinators that help manage congruences,
associativity and sequencing.

For a more interesting example, we prove that the $\zgate$ gate sends
$\sem{\cm{minus}}$ to $\sem{\cm{plus}}$:

\agdafrag{Reasoning.agda}{76}{90}

\noindent The execution equations allow us to extend this reasoning to programs
involving control, \textit{e.g.}:

\agdafrag{Reasoning.agda}{92}{110}

\bigskip We can string together more involved proofs to establish
identities such as the following (where the classical proof of
\agda{xcx} is elided):

\begin{center}
  \begin{tikzpicture}
    \node (l) {
    \begin{tikzpicture} 
    \begin{yquant*}
        qubit {} q0;
        qubit {} q1;
        box {$Z$} q1;
        box {$H$} q1;
        cnot q1 | q0;
    \end{yquant*}
    \end{tikzpicture}
    }; 
    
    \node[right=3mm of l] (eq) {$=$};
    
    \node[right=3mm of eq] (r) {
    \begin{tikzpicture} 
    \begin{yquant*}
        qubit {} q0;
        qubit {} q1;
        box {$Z$} q1 | q0;
        box {$H$} q1;
        not q1;
    \end{yquant*}
    \end{tikzpicture}
    };
  \end{tikzpicture}
\end{center}

\agdafrag{Reasoning.agda}{113}{138}

\subsection{Modeling Complex Numbers}

The $S$ gate is expressed in the computational basis using the matrix
$(\begin{smallmatrix} 1 & 0 \\ 0 & i \end{smallmatrix})$. Although
\qpi\ lacks explicit complex numbers, the language, being
computationally universal, can express them by encoding $a+ib$ as
$(\begin{smallmatrix} a & -b \\ b &
  a \end{smallmatrix})$~\cite{aharonov:toffolihadamard,10.5555/2011508.2011515}. Under
this encoding, we can even express the controlled-$S$ gate:

\agdafrag{QPi/Terms.agda}{118}{122}

\noindent The encoding uses an extra qubit to distinguish the real
part from the imaginary part.

\subsection{Postulating Measurement}

\citet{heunenkaarsgaard:qie} derive quantum measurement as a
computational effect layered on top of a language of unitaries, by
extending the language with effects for \emph{classical cloning} and
\emph{hiding}. Since \qpi\ already has a notion of classical cloning
(two distinct such, in fact) given by the $\cm{copy}_Z$ and
$\cm{copy}_X$ combinators, we only need to extend it with hiding to
obtain measurement of qubit systems.

We can extend \qpi\ with hiding using the exact same arrow
construction as the one used to introduce hiding in
$\mathcal{U}\Pi^\chi_a$ by \citet{heunenkaarsgaard:qie}, with two
subtle differences. The first difference is that, since the model of
\qpi\ is a (dagger) symmetric monoidal category and not a rig
category, this will yield a mere \emph{arrow} over \qpi, and not an
\emph{arrow with choice}:
$$
\frac{c \of b_1 \arzh b_2 \times b_3 \quad b_3 \text{ inhabited}}{\cm{lift}~c
\of b_1 \leadsto b_2}
$$ All available arrow combinators, including the crucial
$\cm{discard} \of b \leadsto 1$ and the derived projections $\cm{fst}
\of b_1 \times b_2 \leadsto b_1$ and $\cm{snd} \of b_1 \times b_2
\leadsto b_2$, are defined precisely as \citet{heunenkaarsgaard:qie}
define them. The second difference concerns partiality of the
model. Since the model of \qpi\ is one of \emph{partial} maps (whereas
the model of $\mathcal{U}\Pi_a$ is one of \emph{total} maps), we need
to accommodate for this in the categorical model. This is precisely
what is done by the $L_\otimes^t$-construction of
\citet{andresmartinezheunenkaarsgaard:universal}.  In short, the
resulting model will not satisfy $\sem{c \ggg \cm{discard}} =
\sem{\cm{discard}}$ for \emph{all} programs $c$, though it will
satisfy it those $c$ for which $\sem{c \ggg \cm{inv}~c} =
\sem{\cm{id}}$.

With this notion of hiding, we can derive measurement in the two bases as
\begin{equation}\label{eq:copyfirstcommute}
  \cm{measure}_Z = \cm{copy}_Z \ggg \cm{fst} \qquad\text{and}\qquad
\cm{measure}_\phi = \cm{copy}_X \ggg \cm{fst}
\end{equation}
exactly as done in previous
work~\cite{coeckeperdrix:envclassicalchannels, heunenkaarsgaard:qie}.
Note that, by commutativity of copying, we could equivalently have
chosen the second projection $\cm{snd}$ instead of $\cm{fst}$. This
can all be expressed in the Agda formalisation as follows:

\agdafrag{QPi/Measurement.agda}{24}{27}

This postulate is \emph{dangerous}, as it does not enforce that it is only applied to
total maps (though we are careful to only do so in the examples here). We hope
to tighten up this loophole in future work.

\agdafrag{QPi/Measurement.agda}{28}{36}

From just this observation, we can show that measurement in the
$\phi$-basis is nothing more than measurement in the $Z$-basis conjugated
by $\hgate$ because:

\agdafrag{Reasoning.agda}{140}{160}

\noindent Following the same principle, we can define measurement in
more exotic bases. For example, measurement in the 2-qubit Bell basis
can be defined by conjugating a pair of $Z$-measurements by the
unitary $\cxgate \ggg \hgate \arrprod \cm{id}$.

\begin{figure}[t]
\begin{center}
\begin{tikzpicture}[scale=1.0,every label/.style={rotate=40, anchor=south west}]
\begin{yquant*}[operators/every barrier/.append style={red, thick}]
    qubit {$a_0=\ket0$} a0;
    qubit {$a_1=\ket0$} a1;
    qubit {$b_0=\ket0$} b0;
    qubit {$b_1=\ket0$} b1;
    box {$H$} a0;
    box {$H$} a1;
    cnot b0 | a0;
    cnot b1 | a0;
    cnot b0 | a1;
    cnot b1 | a1;
    align -;
    box {$H$} a0;
    box {$H$} a1;
    output {$m_0$} a0;
    output {$m_1$} a1;
    output {$r_0$} b0;
    output {$r_1$} b1;
\end{yquant*}
\end{tikzpicture}
\end{center}
\caption{\label{fig:simon}Given a 2-1 function $f : \bool^n
  \rightarrow \bool^n$ with the property that there exists an $a$ such
  that $f(x) = f(x~\textsc{xor}~a)$ for all $x$, Simon's problem is to
  determine $a$. The figure gives the textbook quantum algorithm for
  instance with $n=2$.}
\end{figure}

\subsection{Quantum Algorithms: Simon and Grover}

The language \qpi\ can easily model textbook quantum algorithms. The
circuit in Fig.~\ref{fig:simon} that solves an instance of Simon's
problem can be transliterated directly:

\agdafrag{QPi/Terms.agda}{95}{99}

\noindent The four \cxgate-gates, and more generally an arbitrary
quantum oracle consisting of classical gates, can be implemented in
the underlying classical language and lifted to \qpi.

Having access to measurement allows us to express end-to-end
algorithms as we illustrate with an implementation of a small instance
of Grover's search~\cite{grover:fast}, which, with high probability,
is able to find a particular element $x_0$ in an unstructured data
store of size $n$ by probing it only $O(\sqrt{n})$ times. The
algorithm works by preparing a particular quantum state, and then
repeating a subprogram---the \emph{Grover iteration}, consisting of an
\emph{oracle} stage and an \emph{amplification} stage---a fixed number
of times proportional to $\sqrt{n}$, before finally measuring the
output. The data store of size $n$ is implemented as a unitary $U :
[\bool^n] \to [\bool^n]$ such that $U \ket{x} = -\ket{x}$ if $\ket{x}$
is the element $\ket{x_0}$ being searched for, and $U \ket{x} =
\ket{x}$ otherwise. Though this uses nontrivial phases, it is still a
classical program in disguise, and one could also use a classical
function instead (though this presentation is slightly more
economical). We assume that this unitary is given to us in the form of
a \qpi\ program. The final part of the algorithm is the amplification
stage, which guides the search towards $\ket{x_0}$. This part is the
same for every oracle, depending only on the number of qubits. On
three qubits, the amplifier is given by the circuit
$$
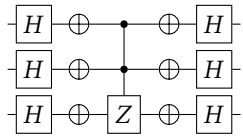
\begin{tikzpicture}[scale=1.0,every label/.style={rotate=40, anchor=south west}]
\begin{yquant*}[operators/every barrier/.append style={red, thick}]
    qubit {} q0;
    qubit {} q1;
    qubit {} q2;
    box {$H$} q0;
    box {$H$} q1;
    box {$H$} q2;
    not q0;
    not q1;
    not q2;
    box {$Z$} q2 | q0, q1;
    not q0;
    not q1;
    not q2;
    box {$H$} q0;
    box {$H$} q1;
    box {$H$} q2;
\end{yquant*}
\end{tikzpicture}
$$

\noindent Using the fact that the $Z$ gate is actually negation
conjugated by Hadamard, we see that this $3$-qubit amplifier is
expressible as the \qpi\ program:

\agdafrag{QPi/Terms.agda}{103}{110}

To put this together, we suppose that we are given a unitary of the
form $u \colon \bool^3 \isozh \bool^3$ described above. The initial
state before iteration should be $\ket{+++}$, and we need to repeat
the Grover iteration $\lceil \tfrac{\pi}{4} \sqrt{2^3} \rceil = 3$
times before measuring the output in the computational basis. All
together, this yields the following \qpi\ program implementing
$3$-qubit Grover search:

\agdafrag{QPi/Measurement.agda}{38}{41}

\section{Conclusion}
\label{sec:conclusion}

We have shown that a computationally universal quantum programming
language arises as a formal combination of two copies of a classical
reversible one, exploiting \emph{complementarity} to guarantee
expressivity. This construction was given as a series of arrows,
providing a positive answer to the existence of a ``quantum
effect.'' Semantically, every step was given as a categorical
construction, giving each language along the way a categorical
semantics. Additionally, we showed that the language can be extended
further with quantum measurements, and that the laws of monoidal
categories, extended with the laws of complementarity, can be used to
reason about quantum programs. The concrete semantics is fully
implemented in an Agda package that was used to run several
examples. The Agda type system can be readily used to verify some
elementary laws (\textit{e.g.}, \agda{inv zero ≡ assertZero}).

We envision the following avenues for future research.

\paragraph{The more $\PiLang$, the better precision}
We have shown that two copies of a classical language, when aligned
just right, are sufficient to yield computationally universal quantum
computation. However, encoding general rotation gates, such as the
ones needed for the quantum Fourier transform, is awkward and
inefficient using just Toffoli and Hadamard. Can additional copies of
the classical base language, when aligned carefully, significantly
improve this---and, if so, by how much?

\paragraph{Completeness}
An equational theory is \emph{sound and complete} for a semantic
domain when any two objects in the domain are equal iff they can be
proven to be equal using the rules of the equational theory. While it
is clear that reasoning in \qpi\ is sound, completeness is wholly
unclear. At the moment, it is unclear how to even proceed, since
already a complete equational theory for the Clifford+T gate set beyond
$2$ qubits is unknown.

\paragraph{Formal quantum experiments}
In investigations by
\citet{abramskyhorsman:demonic,nurgayievamathisdelriorenner:thoughtexp},
programming languages have served as tools in providing outcomes to
thought experiments about physical theories. We envision using (an
extension of) \qpi\ for similar purposes---\textit{e.g.}\ to formulate
contextuality scenarios and quantum protocols, and use the reasoning
capabilities of \qpi\ to answer questions about them.

\paragraph{Infinite-dimensionality and iteration}
\citet{jamessabry:infeff} consider an extended version of $\PiLang$
called $\Pi^o$ with isorecursive types and iteration in the form of a
\emph{trace operator}. Extending \qpi\ to start from this extended
language, rather than the finite $\PiLang$, is of natural
interest. However, doing so successfully would require answering
fundamental open questions about the nature of infinite-dimensional
quantum computation.

\paragraph{Approximate reasoning}
While the categorical semantics of \qpi\ give a way to reason about
quantum programs, it only allows one to prove that programs are
\emph{exactly} equal. However, since approximation plays a large role
in quantum computation, reasoning that two programs are equal not on
the nose, but up to a given error, seems equally important. How can we
extend the model of \qpi\ to account for such approximate reasoning?

\subsection*{Acknowledgements}
We are grateful to Tiffany Duneau for her comments and suggestions
relating to the canonicity theorem.


\bibliography{bibliography.bib}

\end{document}